\documentclass[10pt,orivec,anonymous=true]{llncs}
\usepackage{graphicx}
\usepackage{amssymb}
\usepackage{amsmath}
\usepackage{dcolumn}% Align table columns on decimal point
\usepackage{bm, color}% bold math
\usepackage{graphicx}
\usepackage{diagbox} 
\usepackage{subfig}
\usepackage{caption}
\usepackage{algpseudocode}
\usepackage{algorithmicx}
\usepackage{algorithm}
\usepackage{makecell}
\newcommand{\True}{\textbf{true}}
\newcommand{\False}{\textbf{false}}

\renewcommand{\epsilon}{\varepsilon}

\newcommand{\norm}[1]{\|#1\|}

\renewcommand{\dh}{\ensuremath{\mathcal{D(H)}}}

\DeclareMathOperator{\id}{id}

\renewcommand{\c}{\ensuremath{\mathcal{C}}}
\newcommand{\e}{\ensuremath{\mathcal{E}}}
\newcommand{\f}{\ensuremath{\mathcal{F}}}

\newcommand{\h}{\ensuremath{\mathcal{H}}}

\renewcommand{\l}{\ensuremath{\mathcal{L}}}

 \newcommand{\tr}{{\rm tr}} %trace
 \renewcommand{\a}{\ensuremath{\mathcal{A}}}

\newcommand{\bra}[1]{\langle #1 |}
\newcommand{\ket}[1]{| #1 \rangle}
\newcommand{\braket}[2]{\langle #1 | #2 \rangle}
\newcommand{\ketbra}[2]{| #1 \rangle \langle #2 |}

\newcommand{\cE}{\ensuremath{\mathcal{E}}}

\newcommand{\cH}{\ensuremath{\mathcal{H}}}
\newcommand{\n}{\ensuremath{\mathcal{N}}}
\begin{document}
\title{Robustness Verification of Quantum Classifiers}
%
%\titlerunning{Abbreviated paper title}
% If the paper title is too long for the running head, you can set
% an abbreviated paper title here
%
\author{}
\institute{}

\author{Ji Guan\inst{1} \and
Wang Fang\inst{1,2} \and
Mingsheng Ying\inst{3,1,4}}
\authorrunning{J. Guan et al.}
% First names are abbreviated in the running head.
% If there are more than two authors, 'et al.' is used.
%
\institute{State Key Laboratory of Computer Science, Institute of Software, Chinese Academy of Sciences, Beijing 100190, China \and
University of Chinese Academy of Sciences, Beijing 100049, China \and 
Center for Quantum Software and Information,
University of Technology Sydney, NSW 2007, Australia
 \and
Department of Computer Science and Technology, Tsinghua University, Beijing 100084, China}

\maketitle              % typeset the header of the contribution
\pagestyle{plain}
\begin{abstract}
Several important models of machine learning algorithms have been successfully generalized to the quantum world, with potential speedup to training classical classifiers and applications to data analytics in quantum physics that can be implemented on the near future quantum computers. However, quantum noise is a major obstacle to the practical implementation of quantum machine learning. In this work, we define a formal framework for the robustness verification and analysis of quantum machine learning algorithms against noises. A robust bound is derived and an algorithm is developed to check whether or not a quantum machine learning algorithm is robust with respect to quantum training data. In particular, this algorithm can find adversarial examples during checking. Our approach is implemented on Google's TensorFlow Quantum and can verify the robustness of quantum machine learning algorithms with respect to a small disturbance of noises, derived from the surrounding environment. The effectiveness of our robust bound and algorithm is confirmed by the experimental results, including quantum bits classification as the ``Hello World" example, quantum phase recognition and cluster excitation detection from real world intractable physical problems, and the classification of MNIST from the classical world. 

\end{abstract}
\section{Introduction}
In the last few years, the successful interplay between machine learning and quantum physics shed new light on both fields. On the one hand, machine learning has been dramatically developed to satisfy the need of the industry over the past two decades. At the same time, many challenging quantum physical problems have been solved by automated learning. Notably, inaccessible quantum many-body problems have been solved by neural networks, one instance of machine learning~\cite{carleo2017solving}. On the other hand, as the new model of computation under quantum mechanics, quantum computing has been proved that it can (exponentially) speed up classical algorithms for some important problems~\cite{nielsen2010quantum}. This motivates the development of quantum machine learning and provides the possibility of improving the existing computational power of machine learning to a new level (see the review papers~\cite{biamonte2017quantum,dunjko2018machine} for the details). After that, quantum machine learning was integrated into solving real world problems in quantum physics. One essential example is that quantum convolutional neural networks inspired by machine learning were proposed to implement quantum phase recognition~\cite{cong2019quantum}. Quantum phase recognition asks whether a given input quantum state belongs to a particular quantum phase of matter. At the same time, more provable advantages of quantum machine learning than the classical counterpart have been reported. For instance, the training complexity of quantum models has an exponential improvement on certain tasks~\cite{huang2101information}. Stepping into industries,  Google recently built up a framework \textit{TensorFlow Quantum} for the design and training of quantum machine learning within its famous classical machine learning platform--- TensorFlow~\cite{broughton2020tensorflow}.

Even though quantum machine learning outperforms the classical counterpart in some way, the difficulties in the classical world are expected to be encountered in the quantum case. Classical machine learning has been found to be vulnerable to intentionally-crafted adversarial examples  (e.g.~\cite{huang2011adversarial,goodfellow2014explaining}). Adversarial examples are inputs to a machine learning algorithm that an attacker has crafted to cause the algorithm to make a mistake. One essential mission of machine learning is to prove the absence of or detect adversarial examples used in the defense strategy---adversarial training~\cite{madry2017towards}---appending adversarial examples to the training dataset and retraining the machine learning algorithm to be robust to these examples. However, this goal is not easily achieved~\cite{carlini2017adversarial}. The machine learning community has developed several interesting ideas on designing speciﬁc attack algorithms (e.g. \cite{brown2017adversarial,madry2017towards}) to generate adversarial examples, which is far from measuring the robustness against any adversary. Recently, the formal method community has taken initial steps in this direction~\cite{tran2020verification,elboher2020abstraction,fremont2020formal,kwiatkowska2019safety}, by verifying the robustness of classical machine learning algorithms in a provable way: either a formal guarantee that the algorithms are robust for a given input or a counter-example (adversarial example) is provided if an input is not robust. Some tools have been developed, such as VerifAI~\cite{dreossi2019verifai} and NNV~\cite{tran2020nnv}.
This phenomenon of the vulnerability is more common in the quantum world since quantum noise is inevitable in quantum computation, at least in the current NISQ (Noisy Intermediate-Scale  Quantum) era, and thus led to a series of recent works on quantum machine learning robustness against specific noises. For example, Lu et al.~\cite{lu2019quantum} studied the robustness to various classical adversarial attacks; Du et al.~\cite{du2020quantum} proved that appending depolarization noise in quantum circuits for classifications, a robust bound against adversaries can be derived; Liu and Wittek~\cite{liu2020vulnerability} gave a robust bound for the quantum noise coming from a special unitary group. Very recently, Weber et al.~\cite{weber2020optimal} formalized a link between binary quantum hypothesis testing~\cite{helstrom1967detection} and robust quantum machine learning algorithms for classification tasks. 

Up to our best knowledge, the existing studies of quantum machine learning robustness only consider the situation of a \textit{known} noise source. 
However, a fundamental difference between quantum and classical machine learning is that the quantum attacker is usually the surroundings instead of humans in the classical case, and the information of the environment is unknown. To protect against an \textit{unknown}  adversary, we need to derive a robust guarantee against a worst-case scenario, from which the commonly-assumed known noise sources (e.g.~depolarization noise~\cite{du2020quantum}) are usually far.  Yet in the case of unknown noise, several basic issues are still unsolved: \begin{itemize}\item In theory, it is unclear how to compute a tight and even the optimal bound of the robustness for any given quantum machine learning algorithm. 
\item In practice, an efficient way to find an adversarial example, which can be used to retraining the algorithm to defense the noise, is lacking. Indeed, we do not even know which metric is a better choice measuring the robustness against noise, the same as the classical case against human attackers~\cite{sharif2018suitability}.\end{itemize}

In this work, we define a formal framework for the robustness verification and analysis of quantum machine learning algorithms against noises in which the above problems can be studied in a principled way. 
More specifically, we choose to use fidelity as the metric measuring the robustness as it is one of the most widely used quantity to quantify the uncertainty of noise in the process of quantum computation, and commonly used in quantum engineering and experimental communities (e.g.~\cite{roque2021engineering,torosov2011smooth}). Based on this, an analytical robust bound for any quantum machine learning classification algorithm is obtained and can be applied to approximately checking the robustness of quantum machine learning algorithms. Furthermore, we show that computing the optimal robust bound can be reduced to solving a Semidefinite Programming (SDP) problem. These results lead to an algorithm to exactly and efficiently check whether or not a quantum machine learning algorithm is robust with respect to the training data. A special strength of this algorithm is that it can identify useful new training data (adversarial examples) during checking, and these data can be used to implement adversarial training as the same as classical robustness verification.   The effectiveness of our robust bound and algorithms is confirmed by the case studies of quantum bits classification as the ``Hello World" example of quantum machine learning algorithms, quantum phase recognition and cluster excitation detection from real world intractable physical problems, and the classification of MNIST from the classical world.  

In summary, the main technical contributions of the paper are as follows.

\begin{itemize}
    \item  \textit{Computing the optimal robust bound} of quantum machine classification algorithms is reduced to an SDP (Semidefinite Programming) problem;
    \item \textit{An efficient algorithm} to check the robustness of quantum machine learning algorithms and detect adversarial examples is developed;
    \item The \textit{implementation} of the robustness verification algorithm on Google's TensorFlow Quantum;
    \item \textit{Case studies} -- Checking the robustness of several popular quantum machine learning algorithms for quantum bits classification, cluster excitation detection and the classification  of  MNIST (which are all implemented in Google's TensorFlow Quantum), and quantum phase recognition. 
\end{itemize}
\section{Quantum Data and Computation Models}\label{preliminary}
For convenience of the reader, in this section, we recall some basic concepts of quantum data (states) and the quantum computation model.

The basic data of classical computers are bits, represented by two digits $0$ and $1$. In quantum computing, quantum bits (qubit) play the same role. A qubit is expressed by a normalized complex vector $\ket{\phi}=\left(\begin{array}{cc}a\\ b\end{array}\right)=a\ket{0}+b\ket{1}$
with complex numbers $a$ and $b$ satisfying the  normalization condition $|a|^2+|b|^2=1$. Here, $\ket{0}=\left(\begin{array}{cc}1\\ 0\end{array}\right)$, $\ket{1}=\left(\begin{array}{cc}0\\ 1\end{array}\right)$ correspond to bits $0,1$ respectively, and $\{\ket{0}$, $\ket{1}\}$ is an orthonormal basis of a 2-dimensional Hilbert (linear) space. In general, for a quantum computer consisting of $n$ qubits, a quantum datum is a normalized complex vector $\ket{\psi}$ in a $2^n$-dimensional Hilbert space $\h$. Such a $\ket{\psi}$ is usually called a pure state in the literature of quantum computation.

As a model for computation, a quantum circuit consists of a sequence of, say $m$ quantum logic gates. Each quantum gate can be mathematically represented by a unitary matrix $U_i$ on $\h$, i.e., $U_i^\dagger U_i=U_iU_i^\dagger=I$, where $U_i^\dagger$ is the conjugate transpose of $U_i$ and  $I$ is  the identity matrix on $\h$. Then the circuit is represented by the unitary matrix $U= U_m\cdots U_1$.  
If the quantum datum $\ket{\psi}$ is input to the circuit, then the output is a quantum datum: 
\begin{equation}\label{eq:unitary}
    \ket{\psi'} = U\ket{\psi}.
\end{equation}
%Here, in the computation $U$, $\ket{\psi}$ and $\ket{\psi'}$ are the input and output data, respectively.

In practice, 
%unlike the classical case, 
a quantum datum may not be completely known and can be thought of as a mixed state or ensemble $\{(p_{k},\ket{\psi_{k}})\}_{k}$, meaning that it is at $\ket{\psi_{k}}$ with probability $p_{k}$. Mathematically, it can be described by a density operator $\rho$ (Hermitian  positive semidefinite matrix with unit trace\footnote{$\rho$ has unit trace if $\tr(\rho) = 1$, where trace $\tr(\rho)$ of $\rho$ is defined as the summation of diagonal elements of $\rho$.}) on $\h$: 
\begin{equation}\label{preliminary-1}
    \rho = \sum_{k} p_{k}\ket{\psi_{k}}\bra{\psi_{k}},
\end{equation}
where  $\bra{\psi_{k}}$ is the conjugate  transpose of $\ket{\psi_{k}}$, i.e., $\ket{\psi_{k}}=\bra{\psi_{k}}^\dagger$. In this case, the model of quantum computation is tuned to be a super-operator $\e$, i.e. a mapping from matrices to matrices. It can be written as 
\begin{equation}\label{Eq:evolution}
    \rho'=\e(\rho)
\end{equation}
Here, $\rho$ and $\rho'$ are the input and output data (mixed states) of quantum computation $\e$, respectively. Not every super-operator $\e$ is meaningful in physics. It is required to satisfy the following conditions:
\begin{itemize}
    \item $\e$ is trace-preserving: $\tr(\e(\rho)) = \tr(\rho)$ for all mixed state $\rho$ on $\h$;
    \item $\e$ is completely positive: for any Hilbert space $\h'$, the trivially extended operator $\id_{\h'} \otimes \e$ maps density operators  to density operators on $\h' \otimes \h$, where $\otimes$ denotes the tensor product and $\id_{\h'}$ is the identity map on $\h'$.
\end{itemize}
Such a super-operator $\e$  admits a Kraus matrix form~\cite{nielsen2010quantum}: there exists a set of matrices $\{E_k\}_k$ on $\h$ such that 
$$\e(\rho)=\sum_{k}E_k\rho E_k^\dagger.$$
Here $\{E_k\}_k$ is called Kraus matrices of $\e$.

The behind dynamics  of quantum computers is governed by quantum mechanics, which is applied at the microscopic scale (near or less than $10^{-9}$ meters). At this level, we cannot directly readout the quantum data as the same to the classical counterpart. The only way to extract information from it is  through a quantum measurement, which is mathematically modeled by a set $\{M_{k}\}_{k=1}^m$ of matrices on its state (Hilbert)  space $\h$ with $\sum_{k} M_{k}^\dagger M_{k}=I$. This observing process is probabilistic: if the system is currently in state $\rho$, then a measurement outcome $k$ is obtained with probability  
\begin{equation}\label{Eq:measure_probability}
    p_{k}=\tr(M_{k}^\dagger M_{k}\rho).
\end{equation}
After the measurement, the system's state will be collapsed (changed), depending on the measurement outcome $k$, which is vitally different to the classical computation. 
If the outcome is $k$, the post-measurement  state becomes
\begin{equation}\label{Eq:measure_output}
    \rho'_k=\frac{M_{k}\rho M_{k}^\dagger}{\tr(M_{k}^\dagger M_{k}\rho)}.
\end{equation}
This special property makes it hard to accurately estimate  the distribution $\{p_k\}_k$ unless enough many copies of $\rho$ are provided.

In summary, quantum data have two different forms --- pure state $\ket{\psi}$ and mixed state $\rho$ corresponding to the computation model as a unitary matrix $U$ or a  super-operator $\e$, respectively. Not surprisingly, the latter is a generalization of the former by putting:
\[\rho=\ketbra{\psi}{\psi}, \qquad \qquad \e(\rho)=U\rho U^\dagger.\]
Because of this, the results obtained for mixed states $\rho$ can also be applied to pure states $\ket{\psi}$. Thus, in this paper, we mainly consider mixed states as the quantum data and super-operators as the model of quantum computation.

\section{Quantum Classification  Algorithms}
In this section, we briefly recall quantum classification algorithms. They are designed for \textit{classification of quantum data}.  Essentially, they share the same basic ideas with their classical counterparts but deal with quantum data in the quantum computation model. 
\subsection{Basic Definitions}
In this paper, we focus on a specific learning model called quantum supervised classification. Given a Hilbert space $\h$, we write $\dh$ for the set of all (mixed) quantum states on $\h$ (see its definition in Eq. (\ref{preliminary-1})).
\begin{definition}
  A quantum %(machine learning) 
  classification algorithm $\a$ is a mapping  $\dh\rightarrow\c$, where $\c$ is the set of classes we are interested in.
\end{definition}

Following the training strategy of classical machine learning, 
the classification $\a$ is learned through a dataset $T$ instead of being pre-defined. This training dataset $T=\{(\rho_i,c_i)\}_{i=1}^N$ consists of $N<\infty$ pairs  $(\rho_i,c_i)$, meaning that quantum state $\rho_i$ belongs to  class $c_i$. To learn $\a$, we initialize a \emph{quantum learning model}---a parameterized quantum circuit (including measurement control) $\e_{\theta}$ and a measurement $\{M_k\}_{k\in \c}$. Mathematically, the circuit can be modelled as a quantum super-operator   $\e_{\theta}$ (see its definition in Eq.(\ref{Eq:evolution})),   and  $\theta$ is a set of free parameters that can be tuned.
%For the quantum $K$-multiclass classification task with quantum input state $\rho$ we can employ a quantum circuit to compute $p(\e(\rho))$ instead of using a classical circuit. 
Then for each $k\in \c$, we can compute the probability of the  measurement outcome being $k$:
\begin{equation}\label{Eq:statistics}
	f_k(\theta, \rho) = \tr(M_{k}^\dagger M_{k}\e_{\theta}(\rho)).
\end{equation}
%where $\o=\sum_{\lambda_k}M_{k}^\dagger M_k$, $\lambda_k$ is the eigenvalues of $\o$ and $M_{k}^\dagger M_k$ is the projector onto the corresponding eigenspace. 
It is worth noting that, as we mentioned before, measuring quantum state $\rho$ is probabilistic and $\rho$ will be changed after measuring. So, in practice, accurately estimating $f_k(\theta, \rho)$ for all $k\in\c$  requires enough many copies of $\rho$, which is not the same to the classical case, where a single copy of classical data often meets the training process. 

The quantum classification algorithm $\a$ outputs the class label $c$ for a quantum state $\rho$ using the following condition:
\begin{equation}
	\a(\theta,\rho)=\arg\max_k\tr(M_{k}^\dagger M_{k}\e_\theta(\rho)).
\end{equation}
The learning is carried out as $\theta$ is optimized to minimize the empirical risk
\begin{equation}\label{eq:learning}
  \min_{\theta}\frac{1}{N}\sum_{i=1}^N \l(f(\theta,\rho_i),c_i),
\end{equation}
where $\l$ refers to a predefined loss function, $f(\theta,\rho)$ is a probability vector with each $f_k(\theta, \rho), k\in\c$ as its element, and $c_i$ is also seen as a probability vector with the entry corresponding to $c_i$ being $1$ and others being $0$. The goal  is to find the optimized parameters $\theta^*$  minimizing the risk in Eq.(\ref{eq:learning}) for the given  dataset $T$. Mean-squared error (MSE) is the most popular instance of the empirical risk, i.e., the loss function $\l$ is squared error:
\[\l(f(\theta,\rho_i),c_i)=\frac{1}{C}\lVert f(\theta,\rho_i)-c_i\rVert_2^2,\]
where $C$ is the number of classes in $\c$, and  $\lVert\cdot\rVert_2$ is the $l_2$-norm.

As one can see in the above learning process, the main differences between classical and quantum machine learning algorithms are the learning models and data.

In this paper, we focus on the well-trained quantum classification algorithm $\a$, usually called a quantum classifier.  Here, $\a$ is said to be well-trained if training and   validation accuracy are both high ($\geq 95\%$). The training (validation) accuracy is the frequency that $\a$ successfully classifies the data in a training (validation) dataset. A validation dataset is mathematically equivalent to  a training dataset but only for testing $\a$ rather than learning $\a$.  In this context, $\theta^*$ is naturally omitted, i.e., $\a(\rho)=\a(\theta^*,\rho )$ and $\e(\rho)=\e_{\theta^*}(\rho)$. Briefly, $\a$ only consists of a super-operator $\e$ and a measurement $\{M_{k}\}_k$, denoted by $\a=(\e,\{M_{k}\}_k)$.

\subsection{An Illustrative Example}Let us further illustrate the above definitions by a concrete example---Quantum Convolutional Neural Networks (QCNNs)~\cite{cong2019quantum},  one of the most popular and successful quantum learning models. QCNN extends the main features and structures of the   Convolutional Neural Networks (CNNs) to quantum computing.   
\begin{figure}
\centering
    \includegraphics[width=0.8\linewidth]{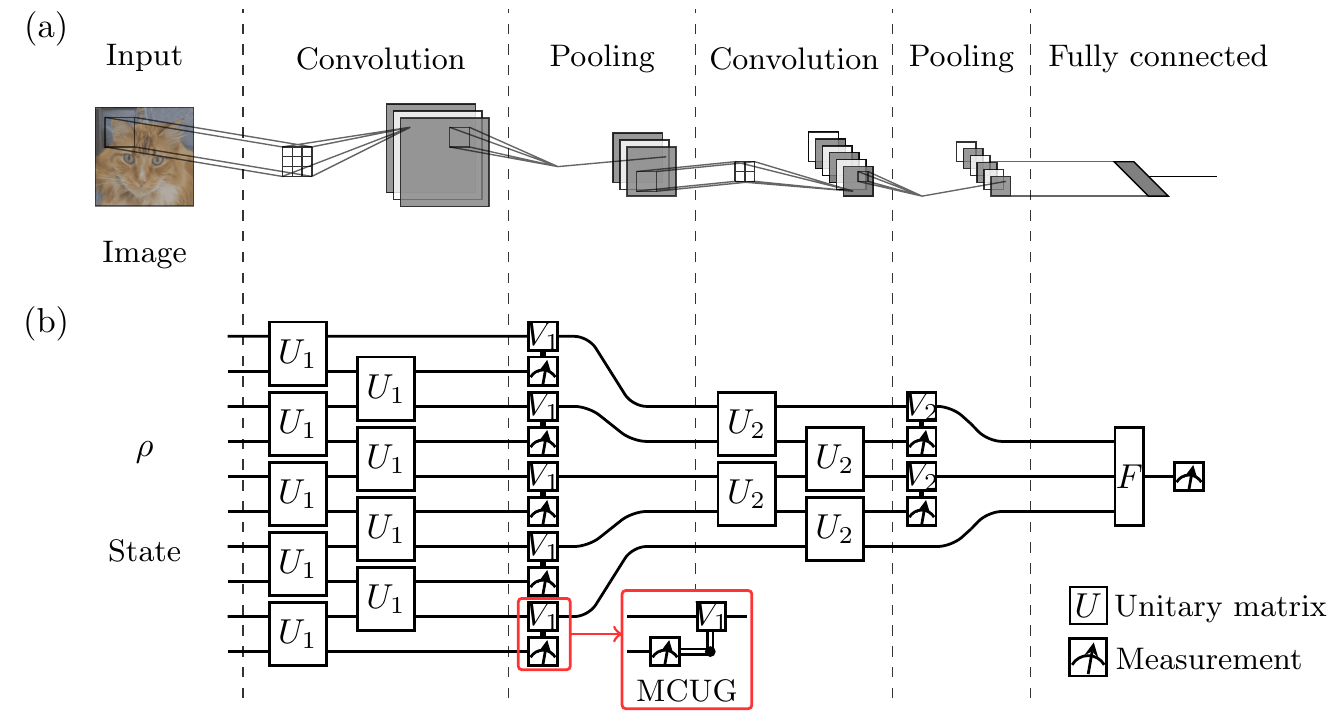}
    \caption{Simple example of CNN and QCNN. QCNN, like CNN, consists of a convolution layer that finds a new state and a pooling layer that reduces the size of the model. Here, MCUG stands for measurement control unitary gate, i.e.,  unitary matrix $V_1$ is applied on the circuit if and only if the measurement outcome is 1.}
    \label{fig.qcnn}
 \end{figure}
The model of QCNN applies the convolution layer and the pooling layer from CNNs to quantum systems, as shown in Figure~\ref{fig.qcnn}(b). The layout proceeds as follows:
\begin{itemize}
    \item [1] The convolution layer (circuit)  applies multiple qubit gates $U_i$ between adjacent qubits to find a new state;
    \item [2] The pooling layer reduces the size of the quantum system by  measuring a fraction of qubits, and the outcomes determine  unitary $V_i$ applied to nearby qubits;
    \item [3] Repeat the convolution layer and pooling layer defined in 1-2;
    \item [4] When the size of the system is sufficiently small, the fully connected layer is applied as a unitary matrix $F$ on the remaining qubits.
\end{itemize}
The input of QCNNs is an unknown quantum state $\rho_{in}$ and the output is obtained by measuring a fixed number of output qubits. As in the classical case, the learning model (defined as the number of convolution and pooling layers) is fixed, but the involved quantum gates (i.e. unitary matrices) $U_{i}, V_{j}, F$ themselves are learned by the above learning process. 

\begin{remark} 
Quantum machine learning can also be used to do classical machine learning tasks.  Image classification, for example, is one of the most successful applications of  Neural Networks (NNs). 
%Quantum computers have potent advantages in terms of superposition and parallel computation. 
To explore the possible advantage of quantum computing, 
Quantum Neural Networks (QNNs)
have been used to  classify images in~\cite{farhi2018classification,oh2020tutorial}). It is shown that by encoding images to a quantum state $\rho_{in}$, QNNs can achieve high accuracy in image classification. We will present a quantum classifier for the classification of MNIST as an example in the evaluation section.\end{remark}

%In summary, quantum machine learning provides brand new learning models to the machine learning community, and the algorithms are learned on quantum computers with potential exponential speedup. The input data are quantum states, which can also be used to code classical information, leading to attain classical tasks. Thus, quantum machine learning is a quantum generalization of classical machine learning in some way.

%Here, we show that it is highly robust to adversarial perturbations
\section{Robustness}

An important issue in classical machine learning is: how robust is a  classification algorithm to adversarial perturbations. A similar issue exists for quantum classifiers against quantum noise. Intuitively, the robustness of quantum classifier  $\a$ is the ability to make correct classification with a small perturbation to the input states. 
%In other words, there are no adversarial examples.
Then a quantum state  $\sigma$ is considered as an adversarial example if it is similar to a benign state $\rho$,  but $\rho$ is correctly classified and $\sigma$ is classified into a class different from that of  $\rho$. Formally, 

\begin{definition}[Adversarial Example]\label{def-robust}
  Suppose we are given a quantum classifier $\a(\cdot)$, an input example $(\rho,c)$, a  distance metric $D(\cdot,\cdot)$ and a small enough threshold value $\epsilon>0$. Then $\sigma$ is said to be an $\epsilon$-adversarial example of $\rho$ if the following is true
  \[(\a(\rho)=c)\land (\a(\sigma)\not=c)\land (D(\rho,\sigma)\leq \epsilon).\]
\end{definition}

The leftmost condition $\a(\rho) = c$ asserts that $\rho$ is correctly classified, the middle condition $\a(\sigma)\not=c$ means that $\sigma$ is incorrectly classified, and the rightmost condition $D(\rho,\sigma)\leq \epsilon$ indicates that $\rho$ and $\sigma$ are similar (i.e., their distance is small). Sometimes, without any ambiguity, $\sigma$ is called an adversarial example of $\rho$ if $\epsilon$ is preset. Notably, by the above definition,  if $\a$ incorrectly classifies $\rho$, then we do not need to consider the corresponding adversarial examples. This is the correctness issue of quantum classifier $\a$ rather than the robustness issue. Hence, in the following discussions, we only consider the set of all correctly recognized states.

The absence of adversarial examples leads to robustness. 
\begin{definition}[Adversarial Robustness]\label{def:adversary} Let $\a$ be a quantum classifier. Then $\rho$ is $\epsilon$-robust for $\a$ if there is no adversarial example of $\rho$.
\end{definition}

The major problem concerning us in this paper is the following:
\begin{problem}[Robustness Verification Problem] Given a quantum classifier $\a(\cdot)$ and an input example $(\rho,c)$. Check whether or not $\a(\sigma)=c$ for all $\sigma\in \n_{\epsilon}(\rho)$,
where $\n_{\epsilon}(\rho)$ is the $\epsilon$-neighbourhood of $\rho$ as 
\[\n_{\epsilon}(\rho)=\{\sigma\in\dh: D(\rho,\sigma)\leq \epsilon)\}.\]
If not, then an adversarial example (counter-example) $\sigma\in\n_{\epsilon}(\rho)$ is provided.
\end{problem}

Obviously, if $\delta$ is a robust bound for an input example $(\rho,c)$ such that  $\a(\sigma)=c \textrm{ for any state } \sigma \in \n_{\delta}(\rho),$ then for any $\epsilon\leq \delta$ (i.e. $\n_{\epsilon}(\rho)\subseteq \n_{\delta}(\rho)$), there is no $\epsilon$-adversarial example of $\rho$.
It is a challenging problem to compute the optimal robust bound 
$\delta^*=\max\delta$ so that  there is no $\epsilon$-adversarial example if and only if $\epsilon\leq \delta^*$.

The above adversarial robustness of quantum states can be generalized to a notion of robustness for quantum classifiers: 
\begin{definition}[Robust Accuracy]\label{def:robust_ac} Let $\a$ be a quantum classifier. The $\epsilon$-robust accuracy of $\a$ is the proportion of  $\epsilon$-robust states in the training dataset.
\end{definition}

\begin{remark}Here, the robust accuracy is defined with respect to the training dataset. In some applications, the dataset can be chosen as another set of quantum states with correct classifications, such as a validation dataset or a combination of it with the training dataset.\end{remark}
%Let us see a concrete example to compute the robust accuracy.  

The reader should notice that the above definitions of robustness for quantum classifiers are similar to those for classical classifiers. 
%The basic reason is that quantum and classical classifiers are both designed to classify data and against noise from the environment and humans, respectively. The main differences are the type of data and the underlying learning model. 
But an intrinsic distinctness between them comes from the choice of distance 
%push us to carefully choose concrete objectives in quantum machine learning in the following. 
%The precondition of the robustness of states is the successful classification by $\a$, and the states in the training dataset are the only data in hand which can be verified by the condition. Therefore,  it is reasonable that the robustness of the quantum classifier $\a$ is based on it. Even though validation datasets are mathematically the same as training datasets and can also be used to define the robustness, the role of it is to test not to interfere $\a$.
 %up to our knowledge, which is not studied before. 
%We saw in Definition \ref{def-robust} that the robustness depends on the  distance
$D(\cdot,\cdot)$. %So, it is essential to properly choose a metric $D$ meaningful in quantum physics. 
In the classical case, humans play the role of the adversary, and then such a distance should promise that a small perturbation is imperceptible to humans, and vice versa. Otherwise, we cannot take the advantage of machine learning over human's distinguishability. For instance, in image recognition, the distance should reflect the perceptual similarity in the sense that humans would consider adversarial examples generated by it perceptually similar to benign image~\cite{sharif2018suitability}. %Finding such a metric is still an open problem~\cite{sharif2018suitability}. 
In the quantum case, it is essential to choose a distance $D$ that is meaningful in quantum physics.
In this paper, we choose to use the distance:
\[D(\rho,\sigma)=1-F(\rho,\sigma)\]
defined by fidelity  $$F(\rho,\sigma)=[\tr(\sqrt{\sqrt{\rho}\sigma\sqrt{\rho}})]^2.$$ Here $\sqrt{\rho}=\sum_{k}\sqrt{\lambda_{k}}\ketbra{\psi_k}{\psi_k}$ if ${\rho}$ admits the spectral decomposition $\sum_{k}\lambda_{k}\ketbra{\psi_k}{\psi_k}.$
%world, particularly in the NISQ era, quantum computation is attacked by noise from the surrounding environment, and such a metric should measure the uncertainty bought by them.  
Fidelity is one of the most widely used quantities to quantify such uncertainty of noise
by the experimental quantum physics and quantum engineering communities (see e.g.~\cite{myerson2008high,burrell2010scalable}). 
%Thus, we choose to use the metric:

\begin{remark}The  trace distance has been used in recent literature  (e.g.~\cite{du2020quantum}) for some issues related to quantum robustness verification:
\[T(\rho,\sigma)=\frac{1}{2}\|\rho-\sigma\|_{tr}=\frac{1}{2}\tr[\sqrt{(\rho-\sigma)^\dagger(\rho-
\sigma)}].\]
It is a generalization of the total variation distance, which is a distance measure for probability distributions.
So far, to the best of our knowledge, there is no discussion about which distance is better in the literature. Here, we argue that fidelity is better than trace distance in the context of quantum machine learning against quantum noise. As we know, state distinguishability is the basis of measuring the effect of noise on quantum computation. 
The main difference between trace distance $T(\rho,\sigma)$ and fidelity $F(\rho,\sigma)$ is the number of copies of states $\rho$ and $\sigma$ as the resource required in the experiments for distinguishing them. More precisely, 
trace distance quantifies the maximum probability of correctly guessing through a measurement whether $\rho$ or $\sigma$ was prepared, while fidelity asserts the same quantity whence infinitely many samples of $\rho$ and $\sigma$ can be supplied (See Appendix A of the extended version of this paper~\cite{guan2020robustness} for more details). In quantum machine learning, a large enough number of copies of the states are the precondition of statistics in Eq.(\ref{Eq:statistics}) for learning and classification. Thus, fidelity is more suitable than trace distance for our purpose.\end{remark}

\section{Robust Bound}
In this section, we develop a theoretic basis for robustness verification of quantum classifiers. After setting the distance $D$ to be the one defined by fidelity, a robust bound can be derived.
\begin{lemma}[Robust Bound]\label{lem:bound}
       Given a quantum classifier $\a =(\e,\{M_k\}_{k\in\c})$  and a quantum state $\rho$.
       Let $p_1$ and $p_2$ be  the first and second largest elements of $\{\tr(M_{k}^\dagger M_k\e(\rho))\}_k$, respectively.
       If $\sqrt{p_1}-\sqrt{p_2}>\sqrt{2\epsilon}$, then $\rho$ is $\epsilon$-robust. 
\end{lemma}
\begin{proof}
See  Appendix~B of the extended version of this paper~\cite{guan2020robustness}.
\end{proof}

 The above robust bound gives us a quick robustness verification by the measurement outcomes of $\rho$ without searching any possible adversarial examples. Furthermore, it also can be used to compute an under-approximation of the robust accuracy of $\a$ by one-by-one checking the robustness of quantum states in the training dataset. We will see that the robust bound and the induced robust accuracy scales well in the later experiments.   
 However, $\sqrt{p_1}-\sqrt{p_2}>\sqrt{2\epsilon}$ is not a necessary condition of $\epsilon$-robustness. Fortunately, when $\sqrt{p_1}-\sqrt{p_2}\leq \sqrt{2\epsilon}$, we can compute the optimal robust bound by Semidefinite Programming (SDP). %It is worth noting that this is a complete method for the verification of $\epsilon$-robustness. 
 Recall that 
SDP is a convex optimization concerned with the optimization of a linear objective function over the intersection of the cone of positive semidefinite matrices with an affine space. It has the form
\begin{align*}
    \min \quad & \tr(C X) \\
    \text{subject to} \quad & \tr(A_k X) \leq b_k, \quad \text{for } k = 1, \ldots, m \\
    & X \geq 0
\end{align*}
where $C, A_1, \ldots, A_m$ are all Hermitian  $n\times n$ matrices (i.e. $A^\dagger=A$), and $X$ is the optimization variable $n\times n$ matrix with $X\geq 0$, i.e.,  $X$ is positive semidefinite. Many efficient solvers have been developed for solving SDPs---not only compute the minimal value, but also output a corresponding optimal solution $X$. The following two theorems show that that checking $\epsilon$-robustness   and computing optimal robust bound of quantum states can both be reduced to an SDP.
\begin{theorem}[$\epsilon$-robustness Verification]\label{Theo:checkrobust}
Let $\a=(\e,\{M_k\}_{k\in\c})$ be a quantum classifier and $\rho$ be a state with $\a(\rho)=l$. Then $\rho$ is $\epsilon$-robust if and only if for all $k\in\c$ and $k\not=l$, the following problem has no solution (feasibility problem):
\begin{align*}
\min_{\sigma\in \dh}\quad&0\\
\text{subject to}  \quad  &\sigma \geq 0  \\
    &\tr(\sigma)=1 \\
    & \tr[(M_{l}^\dagger  M_l-M_{k}^{\dagger}M_k)\e(\sigma)]\leq 0\\
   & 1-F(\rho,\sigma)\leq \epsilon 
\end{align*}
%Furthermore, the above problem can be reduced to a SDP. 
\end{theorem}
\begin{proof}
See Appendix~C of the extended version of this paper~\cite{guan2020robustness}.
\end{proof}

Actually, the objective function $0$ in the above theorem can be chosen as any constant number.
\begin{theorem}
[Optimal Robust Bound]\label{Theo:SDP}
Let $\a$ and $\rho$ be as in Theorem~\ref{Theo:checkrobust} with $\a(\rho)=l$, and let $\delta_k$ be the solution of the following problem:
\begin{align*}
\delta_k=\min_{\sigma\in \dh}\quad& 1-F(\rho,\sigma)\\
\text{subject to} \quad    &\sigma \geq 0  \\
    &\tr(\sigma)=1 \\
    & \tr[(M_{l}^\dagger  M_l-M_{k}^{\dagger}M_k)\e(\sigma)]\leq 0
\end{align*}
where if the problem is  unsolved, then  $\delta_k=+\infty.$
Then $\delta=\min_{k\not=l}\delta_k$  is the optimal robust bound of $\rho$.
%, where  and Furthermore, the above problem can be reduced to a SDP.
\end{theorem}
\begin{proof}
 The proof is similar to Theorem~\ref{Theo:checkrobust}.
\end{proof}

\begin{remark}
One may wonder why checking $\epsilon$-robustness and computing the optimal robust bound can always be reduced to an SDP. This is indeed implied by the \textit{basic quantum mechanics postulate} of linearity; more specifically, all of the super-operators and measurements used in quantum machine learning algorithms are linear. In contrast, the functions represented by the neural networks in classical machine learning may be nonlinear as the pooling layer is not linear. As a result, the reduced optimization problem for the robustness verification is not convex (e.g. \cite{ruan2018global}).  For overcoming this difficulty, many different methods have been developed to encode the nonlinear activation functions as linear constraints. Examples include NSVerify~\cite{lomuscio2017approach}, MIPVerify~\cite{tjeng2017evaluating}, ILP~\cite{bastani2016measuring} and ImageStar~\cite{tran2020verification}.\end{remark}

\section{Robustness Verification Algorithms}

\begin{algorithm}[htbp]
\caption{StateRobustnessVerifier($\a,\epsilon,\rho,l$)}
\label{Algorithm}
    \begin{algorithmic}[1]
    \Require $\a=(\e,\{M_{k}\}_{k\in\c})$ is a well-trained quantum classifier, $\epsilon < 1$ is a real number, $(\rho,l)$ is an element of the training dataset of $\a$
    \Ensure \True{} indicates $\rho$ is $\epsilon$-robust or \False{} with an adversarial example $\sigma$ indicates $\rho$ is not $\epsilon$-robust  
    \ForAll{$k\in\c$ and $k\not=l$}\label{Algorithm:line:count}
    \State\label{Algorithm:line:SDP} By a SDP solver, compute $\delta_k$ with an optimal state $\sigma_k$ in the SDP of Theorem~\ref{Theo:SDP}
    \EndFor
    \State Let $\delta=\min_k\delta_k$ and $k^*=\arg\min_k\delta_k$
    \If{$\delta>\epsilon$}
    \State \Return \True{}
    \Else
    \State \Return \False{} and $\sigma_{k*}$
    \EndIf
    \end{algorithmic}  
\end{algorithm}

In this section, we develop several algorithms for verifying the robustness of quantum classifiers based on the theoretic results presented in the last section. 

First, let us consider the robustness of a given quantum state $\rho$. In many applications (as shown in our experiments in Section \ref{evaluation}), we are required to check whether  $\rho$ is $\epsilon$-robust for an arbitrarily given threshold $\epsilon$.
Note that 
once we computed  the optimal robust bound $\delta$, checking $\epsilon$-robustness of $\rho$ is equivalent to compare $\epsilon$ and $\delta$; that is, $\epsilon\leq \delta$ if and only if $\rho$ is $\epsilon$-robust. Combining this simple observation with Theorem \ref{Theo:checkrobust}, we obtain %Therefore, following this way, we can efficiently do the robustness verification of $\rho$ for different values of $\epsilon$, which is the situation in the following experiments. We summary this approach in
Algorithm~\ref{Algorithm} for checking the $\epsilon$-robustness of $\rho$ and finding the minimum adversarial perturbation $\delta$ caused by quantum noise.   The main cost of Algorithm~\ref{Algorithm} incurs in solving SDPs in Line~\ref{Algorithm:line:SDP}, which scales as $O(n^{6.5})$ by interior point methods~\cite{zhang2018sparse}, where $n$ is the number of rows of the semidefinite matrix $\rho$ in SDP, i.e., the dimension of Hilbert space of the quantum states in our case. As we need to apply an SDP solver for $|\c|-1$ times in Line~\ref{Algorithm:line:count}, the total complexity is as follows.
\begin{theorem}
The worst case complexity of Algorithm~\ref{Algorithm} is $O(|\c|\cdot n^{6.5})$, where $n$ is the dimension of input state $\rho$ and $|\c|$ is the number of the set $\c$ of classes we are interested in.
\end{theorem}

Now we turn to consider the robustness of a quantum classifier $\a$. Algorithm~\ref{algo:robustcheck} is designed for checking robustness of $\a$ by combining Algorithm~\ref{Algorithm} with Lemma \ref{lem:bound} (see the discussion in the paragraph  after Lemma \ref{lem:bound}). 
\begin{algorithm}[htbp]
\caption{RobustnessVerifier($\a,\epsilon,T$)}
\label{algo:robustcheck}
    %\begin{algorithmic}[1]
    %\Require  $\a=(\e,\{M_k\}_{k\in\c})$ is %a well-trained quantum classifier, %$\epsilon < 1$ is a real number, %$T=\{(\rho_i,l_{i})\}$ is the training %dataset of $\a$
    %\Ensure The  robust accuracy $RA$ and a %set $R=\{<\sigma_j, i_{j}>\}$, where %for each $j$, $\rho_j$ is an %$\epsilon$-adversarial example of %$\rho_{i_j}$; $R$ can be an empty set %if all states in $T$ are %$\epsilon$-robust.
    %\State $R=\emptyset$ be an empty set. %\Comment{Recording adversarial examples %and corresponding indexes of states in %training dataset $T$}
    %\ForAll{$(\rho_i,l_i)\in T$}
    %\State Let $p_1$ and $p_2$ be the first %and second largest elements of %$\{\tr(M_{k}^\dagger %M_k\e(\rho_i))\}_k$, respectively.
    %\If {$\sqrt{p_1}-\sqrt{p_2}>\sqrt{2\eps%ilon}$} \Comment{Applying the robust %bound in Lemma~\ref{lem:bound}}
    %\State \textbf{Skip}
    %\Else
    %\If{StateRobustnessVerifier %$(\a,\epsilon,\rho_i,l_i)$==\True{}} 
    %\State \textbf{Break}
    %\Else
    %\State $\sigma$ be the output state of %StateRobustnessVerifier %$(\a,\epsilon,\rho_i,l_i)$ 
    %\State $R=R\cup \{(\sigma, i)\}$
    %\EndIf
    %\EndIf
    %\EndFor
    %\State \Return $RA=1-\frac{|R|}{|T|}$, %$R$ \Comment{$|R|=0$ if $R$ is a empty %set}
    %\end{algorithmic}
    \begin{algorithmic}[1]
    \Require  $\a=(\e,\{M_k\}_{k\in\c})$ is a well-trained quantum classifier, $\epsilon < 1$ is a real number, $T=\{(\rho_i,l_{i})\}$ is the training dataset of $\a$
    \Ensure The  robust accuracy $RA$ and a set $R=\{<\sigma_j, i_{j}>\}$, where for each $j$, $\rho_j$ is an $\epsilon$-adversarial example of $\rho_{i_j}$; $R$ can be an empty set if all states in $T$ are $\epsilon$-robust.
    \State $R=\emptyset$ be an empty set. \Comment{Recording adversarial examples and corresponding indexes of states in training dataset $T$}
    \ForAll{$(\rho_i,l_i)\in T$}
    \State Let $p_1$ and $p_2$ be the first and second largest elements of $\{\tr(M_{k}^\dagger M_k\e(\rho_i))\}_k$, respectively.
    \If {$\sqrt{p_1}-\sqrt{p_2}\leq\sqrt{2\epsilon}$} \Comment{Applying the robust bound in Lemma~\ref{lem:bound}}
    \If{StateRobustnessVerifier $(\a,\epsilon,\rho_i,l_i) ==$ \False{}} 
    \State $\sigma$ be the output state of StateRobustnessVerifier $(\a,\epsilon,\rho_i,l_i)$ 
    \State $R=R\cup \{(\sigma, i)\}$
    \EndIf
    \EndIf
    \EndFor
    \State \Return $RA=1-\frac{|R|}{|T|}$, $R$ \Comment{$|R|=0$ if $R$ is a empty set}
    \end{algorithmic}
\end{algorithm}
A major benefit of formal robustness verification for  classical classifiers is perhaps that it can be used to detect a counter-example (adversarial example) for a given input~(see e.g.~\cite{tran2020verification,elboher2020abstraction,fremont2020formal,kwiatkowska2019safety}). This benefit is kept in Algorithm~\ref{algo:robustcheck} for the robustness verification of quantum classifiers. In particular, we are able to extend the technique of adversarial training in classical machine learning~\cite{madry2017towards} into the quantum case:  an adversarial example $\sigma$ is automatically generated once $\epsilon$-robustness of $\rho$ fails, and then by appending $(\sigma, l)$ into the training dataset, we can retrain $\a$ to improve the robustness of the classifier.

To analyze the complexity of 
Algorithm~\ref{algo:robustcheck}, we first see by Theorem~\ref{Theo:SDP} that for evaluating the robustness of $\a$---computing its robust accuracy and finding its adversarial examples, one need to call  Algorithm~\ref{Algorithm} for each quantum state in the training dataset, which costs $ O(|\c|\cdot n^{6.5})$. Thus, the total complexity of robustness verification is $O(|T|\cdot |\c|\cdot n^{6.5})$, where $|T|$ is the number of elements in the training dataset $T$. However, the robust bound given in Lemma~\ref{lem:bound} can help to speed up the process by quickly finding all potential non-robust states,  as the complexity of finding the bound is only $O(|\c|\cdot n^3)$, which is the cost of $|\c|$ times of the multiplication of two $n\times n$ matrices. In practice, this bound scales well, as confirmed by our experiments presented in Section \ref{evaluation}. Therefore, a good strategy for implementing the robustness verification is that we first use robust bound to pick up all potential non-robust states from the given training dataset $T$ and store them in a set $T'$. Then we check all left candidates in the training dataset $T$ one-by-one using  Algorithm~\ref{Algorithm} and use a set $R$ to record the found adversarial examples and the corresponding indexes of states. This strategy  can significantly reduce the complexity to $O(|T'|\cdot |\c| \cdot n^{6.5})$.
%for $|T'|\ll m$, where $|T'|$ is the number of $T'$ containing all left states after applying robust bounds.
%The detailed implementation is summarized in Algorithm~\ref{algo:robustcheck} with time complexity $O(|T'|\cdot |\c|\cdot n^{6.5})$. 
Indeed, our experiments show that the robust bound given in Lemma~\ref{lem:bound} scales very well in the sense of  $|T'|\ll |T|$.

\begin{remark}
 Thanks to the linearity of the quantum learning model determined by the basic postulate of quantum mechanics,
 the robustness verification of quantum classifiers
 can be done in an efficient way  (with polynomial time complexity in the size of input state). It is usually not the case in verifying the robustness of classical machine learning algorithms.
% whereas only approximation verification can be efficiently obtained as such good structure is absent, which leads to some hard problems. 
For example, DNNs are often non-linear and non-convex and verifying even some simple properties of them can be an  NP-complete problem~\cite{katz2017reluplex}. 
 
 Surprisingly, the robustness verification problem for quantum classifiers becomes much harder if we are required to find adversarial examples in \textit{pure states}. Roughly speaking, the reason is that the set of all pure states is not convex, and thus computing the optimal robust bound for pure states is not an SDP, as in Theorem~\ref{Theo:SDP}. We can prove that it is a  Quadratically Constrained Quadratic Program (QCQP), an optimization problem where both the objective function and the constraints are quadratic functions (see Appendix~D of the extended version of this paper~\cite{guan2020robustness} for the proof), which is NP-hard.  Algorithm~\ref{Algorithm} can be adapted to this pure state robustness verification by calling a QCQP solver instead of an SDP solver in Line~\ref{Algorithm:line:SDP}. Subsequently,  Algorithm~\ref{algo:robustcheck} can use this new version of Algorithm~\ref{Algorithm} as a subroutine to compute the corresponding robust accuracy and find adversarial examples of pure states. We will evaluate the QCQP-based robustness verification in the case study of MNIST classification in which handwritten digits are encoded in pure states. 

 %But for the case of binary classification, the reduced QCQP is convex and can be efficiently solved using interior point methods, as done with SDPs. See Appendix~\ref{appendix:pure_state_robustness} for the details. 
\end{remark}

\section{Evaluation}\label{evaluation}
Algorithm~\ref{algo:robustcheck} is implemented on \textit{TensorFlow Quantum} --- a platform of Google for designing and training quantum machine learning algorithms, by calling an SDP solver --- CVXPY: Python Software for Disciplined Convex Programming~\cite{diamond2016cvxpy}. This section aims to evaluate our approach with experiments on some concrete examples. This section is arranged as follows. In Subsections \ref{sec:7-1}-\ref{sec:7-4}, we present several    well-trained quantum classifiers. Then the evaluation is carried out in Subsection \ref{sec:7-5} by applying Algorithm~\ref{algo:robustcheck} to check the robustness verification of these classifiers and find their adversarial examples if existing. 

To demonstrate our method as sufficiently as possible, we check the robustness of four quantum classifiers. We begin with a “Hello World” example --- qubits classification, and then we step in two quantum classifiers applied to real world tasks --- quantum phase recognition and cluster excitation detection, which are both fundamental and hard problems in quantum physics. At last, to compare with classical robustness verification, we consider the classification of MNIST by encoding handwritten digital images into quantum data. \textit{These experiments cover all illustrated examples of TensorFlow Quantum}.  

\subsection{Quantum Bits Classification}\label{sec:7-1}
A \textquotedblleft Hello World" example of quantum machine learning is quantum bits classification~\cite{broughton2020tensorflow}. 
The aim is to implement a binary classification for regions on a single qubit, i.e., a perceptron for qubits. Specifically, two random normalized vectors $\ket{a}$ and $\ket{b}$ (pure states) in the $X$-$Z$ plane of the Bloch sphere are chosen. Around these two vectors, we randomly sample two sets of quantum data points; the objective is to learn a quantum gate to distinguish the two sets. 
\begin{figure}
\centering
    \includegraphics[width=0.9\linewidth]{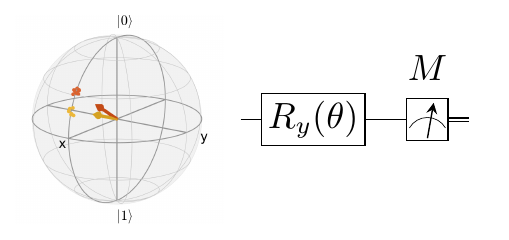}
    \caption{Training model of quantum bits classification: the left figure shows the samples of quantum training dataset represented on the Bloch sphere. Samples are divided into two categories, marked by red and yellow, respectively. The vectors are the states around which the samples were taken. The first part of the right figure is a parameterized rotation gate, whose job is to remove the super-positions in the quantum data. The second part is a measurement $M$ along the Z-axis of the Bloch sphere converting the quantum data into classes.}
    \label{fig.binary.model1}
 \end{figure}
 A concrete instance of this type is shown in Fig. \ref{fig.binary.model1}. 
In this example, the angles with $\ket{0}$ (Z-axis) of the two states $\ket{a}$ and $\ket{b}$ are $\theta_a=1$ and $\theta_b=1.23$, respectively; see the first figure in  Fig.~\ref{fig.binary.model1}. Around these two vectors, we randomly sample two sets (one for category \textquotedblleft a\textquotedblright\  and one for category \textquotedblleft b\textquotedblright) of quantum data points on the sphere, forming a dataset. The dataset consists of $800$ samples for the training and $200$ samples for the validation. 
As shown in Fig.~\ref{fig.binary.model1}, we use a parameterized rotation gate $R_y(\theta) = e^{-i\sigma_y\theta/2}$ and a measurement  $M = \{M_a = \ket{0}\bra{0}, M_b = \ket{1}\bra{1}\}$ to do the  classification. Targeting to minimizing the MSE form of Eq.~(\ref{eq:learning}),
we use Adam optimizer~\cite{KingmaB14} to update $\theta$.  After training, we achieve both $100\%$ training and validation accuracy, and the final parameter $\theta$ is $0.4835$.

%Next, we initialize $\epsilon$-robustness checking for the classifier. We apply our robust bound (here the value is 0.02) in Lemma~\ref{lem:bound} to pick up all potential non-robust states from the 800 points in the training dataset. No points are left, i.e., the $10^{-4}$-robust accuracy of the classifier is 100\%.

\subsection{Quantum Phase Recognition}

Quantum phase recognition (QPR) of one dimensional many-body systems has been attacked by quantum convolutional neural networks (QCNNs) proposed by Cong et al.~\cite{cong2019quantum}. Consider a $Z_2\times Z_2$ symmetry-protected topological (SPT) phase $\mathcal{P}$ and the ground states of a family of Hamiltonians on spin-$1/2$ chain with open boundary conditions:
\[ H = -J\sum_{i=1}^{N-2}Z_iX_{i+1}Z_{i+2} - h_1\sum_{i=1}^N X_i - h_2\sum_{i=1}^{N-1} X_i X_{i+1}\]
where $X_i,Z_i$ are Pauli matrices~\cite{nielsen2010quantum} for the spin at site $i$, and the $Z_2\times Z_2$ symmetry is generated by $X_{\textrm{even}(\textrm{odd})} = \prod_{i\in\textrm{even}(\textrm{odd})}X_i.$ The goal is to identify whether the ground state $\ket{\psi}$ of $H$ belongs to phase $\mathcal{P}$ when $H$ is regarded as a function of $(h_1/J, h_2/J)$. For small $N$,  a numerical simulation can be used to exactly solve this problem~\cite{cong2019quantum}; See Fig.~4a in Appendix~E of the extended version of this paper~\cite{guan2020robustness} for the exact phase boundary points (blue and red diamonds) between SPT phase and non-SPT (paramagnetic or antiferromagnetic) phase  for $N=6$. Thus the $6$-qubit instance is an excellent testbed for different new methods and techniques of QPR. Here, we train a QCNN model to implement $6$-qubit QPR in this setting.

To generate the dataset for  training, we sample a serials of Hamiltonian $H$ with $h_2/J=0$,   uniformly varying $h_1/J$ from $0$ to $1.2$ and compute their corresponding ground states; see the gray line of Fig.~4a in Appendix~E of the extended version of this paper~\cite{guan2020robustness}. For the testing, we uniformly sample a set of validation data from two random regions of the 2-dimensional space $(h_1/J, h_2/J)$; see the two dashed rectangles of Fig.~4a. Finally, we obtain $1000$ training data and $400$ validation data. Our parameterized QCNN circuit is shown in Fig.~4b in Appendix~E of the extended version of this paper~\cite{guan2020robustness}, and the unitaries $U_i, V_j, F$ are parameterized with generalized Gell-Mann matrix basis~\cite{Bertlmann_2008}: {$U = \exp(-i\sum_j \theta_j\Lambda_j)$}, where $\Lambda_j$ is a matrix and $\theta_j$ is a real number; the total number of parameters $\theta_j$, $\Lambda_j$ is 114.  For the outcome measurement of one qubit, we use measurement   $M = \{M_0 = \ketbra{+}{+}, M_1 = \ketbra{-}{-}\}$ to predict that input states belongs to $\mathcal{P}$ with output $0$, where $\ket{\pm}=\frac{1}{\sqrt{2}}(\ket{0}\pm\ket{1})$. Targeting to minimizing the MSE form of Eq.~(\ref{eq:learning}),
we use Adam optimizer to update the $114$ parameters.
After training, $97.7\%$ training accuracy and  $95.25\%$ validation accuracy are obtained. At the same time, our classifier conducts a phase diagram (the colorful figure in  Fig.~4a), where the learned phase boundary almost perfectly matches the exact one gotten by the numerical simulation. All these results indicate that our classifier is well-trained.

\subsection{Cluster Excitation Detection}
The task of cluster excitation detection is to train a quantum classifier to detect if a prepared cluster state is ``excited" or not~\cite{broughton2020tensorflow}. Excitations are represented with a $X$ rotation on one qubit. A large enough rotation is deemed to be an excited state and is labeled by $0$, while a rotation that isn't large enough is labeled by $1$ and is not deemed to be an excited state. Here, we demonstrate this classification task with $6$ qubits. We use the circuit shown in Fig.~5a of Appendix~E in the extended version of this paper~\cite{guan2020robustness} to generate training ($840$) and validation ($360$) samples. The circuit  generates cluster state by performing a $X$ rotation (we omit angle $\theta$) on one quibit. The rotation angle $\theta$ is ranging from $-\pi$ to $\pi$ and if $-\pi/2 \leq \theta\leq \pi/2$, the label of the output state is $1$; otherwise, the label is $0$. The classification circuit model (a quantum convolutional neural network) uses the same structure in TensorFlow Quantum~\cite{broughton2020tensorflow}, shown in Fig.~5b of Appendix~E in the extended version of this paper~\cite{guan2020robustness}. The explicit parameterization of  $C_i,P_j$ can be found in \cite{broughton2020tensorflow}. The final measurement  $M = \{M_0 = \ketbra{0}{0}, M_1 = \ketbra{1}{1}\}$. Targeting to minimizing the MSE form of Eq.~(\ref{eq:learning}),
we use Adam optimizer to update all $C_i,P_j$.  We achieve $99.76\%$ training accuracy and $99.44\%$ validation accuracy.

\subsection{The Classification of MNIST}\label{sec:7-4}
Handwritten digit recognition is one of the most popular tasks in the classical machine learning zoo.    
The archetypical training and validation data come
from the MNIST dataset which consists of 55,000 training samples handwritten digits~\cite{lecun-mnisthandwrittendigit-2010}. These digits have been labeled by humans as representing
one of the ten digits from number 0 to 9, and are in the form of gray-scale images that contains $28\times 28$ pixels. 
 Each pixel has a grayscale value ranging from $0$ to $255$. Quantum machine learning has been used to distinguish a too simplified version of MNIST by downscaling the image sizes to $8\times 8$ pixels. Subsequently, the numbers represented by this version of MNIST can not be perceptually recognized~\cite{broughton2020tensorflow}. Here, we build up a quantum classifier to recognize a MNIST version of $16\times 16$ pixels (see second column images of Fig.~\ref{fig.mnist.robust1}). As demonstrated in~\cite{broughton2020tensorflow}, we select out $700$ images of number $3$ and $700$ images of number $6$ to form our training ($1000$) and validation ($400$) datasets. Then we downscale those $28\times 28$ images to $2^4\times 2^4$ images (fitting the size of quantum data),  and encode them into the pure states of $8$ qubits via amplitude encoding.  Amplitude encoding uses the amplitude of computational basis to represent vectors with normalization:
\[ (x_0, x_2, \dots, x_{n-1}) \to \sum_{i=0}^{n-1}\frac{x_i}{\sqrt{\sum_{j=0}^{n-1}\vert x_j\vert^2}}\ket{i}.\]
where $\{\ket{i}\}$ is a set of orthogonal basis of the 8 qubits state space.  
The normalization doesn't change the pattern of those images. For learning a quantum classifier, we use the QCNN model in Fig.~6 of Appendix~E in the extended version of this paper~\cite{guan2020robustness} and use measurement  $M = \{M_0 = \ketbra{+}{+}, M_1 = \ketbra{-}{-}\}$.  The output of measurement $M$  indicates the numbers: output $1$ for number  $3$ and output $0$ for  number $6$. The explicit parameterization of those $C_i,P_j$ can be found in \cite{broughton2020tensorflow}. Again we use Adam optimizer to update the model parameters minimizing the MSE form of Eq.(\ref{eq:learning}).  We finally achieve $98.4\%$ training accuracy and $97.5\%$ validation accuracy.

%Next, we initialize the $\epsilon$-robustness checking for the classifier. For example, we illustrate the details of the process for the case of $\epsilon=0.0001$. First, we only apply our robust bound in Lemma~\ref{lem:bound} to pick up all potential non-robust states from the 1000 points in the training dataset. Only $3$ points are left and the bound scales well. Thus the under-approximation robust accuracy is  99.70\%. Then, we check $0.0001$-robustness  by Algorithm~\ref{algo:robustcheck}. Indeed, $2$ of the points detected by the above robust bound are non-robust and the exact robust accuracy of the QCNN is 99.80\%. We also compare the verification time of the two approaches to the robust accuracy.  See the second column in  Table~\ref{tab:QPR} for the detail, and other experiment results of $\epsilon$-robustness are also summarized in the same table. The former approach is super faster than the latter approach, and the under-approximation robust accuracy is almost equal to the exact one. 

\subsection{Robustness Verification}\label{sec:7-5}
Now, we start to check the $\epsilon$-robustness for the above four well-trained classifiers presented in the previous four subsections. 

In practical applications, the value of robustness $\epsilon$ in Definition~\ref{def:adversary} represents the ability of state preparation by quantum controls. For example, the state-of-the-art is that a single qubit can be prepared with fidelity 99.99\% (e.g.~\cite{myerson2008high,burrell2010scalable}). Here, we choose four different values of $\epsilon$ in each experiment. 
\begin{table}
     \vspace{-0.6cm}
    \centering
    \resizebox{0.8\linewidth}{!}{
    \renewcommand\arraystretch{1.2}
    \setlength{\arrayrulewidth}{.1em}
    \begin{tabular}{lcccc}
        \cline{2-5}
         & \multicolumn{4}{c}{Robust Accuracy (in Percent)} \\
         \cline{2-5}
         &  $\epsilon = 0.001$ & $\epsilon=0.002$ & $\epsilon = 0.003$ & $\epsilon = 0.004$ \\
         \cline{1-5}
         \makecell[c]{Robust Bound \\ (Lemma~\ref{Theo:checkrobust} - Algorithm~\ref{algo:underapproximation})} & $88.13$ & $75.88$ & $58.88$ & $38.25$ \\
         \makecell[c]{Robustness Algorithm \\ (Theorem~\ref{Theo:SDP} - Algorithm~\ref{algo:robustcheck})} & $90.00$ & $76.50$ & $59.75$ & $38.88$ \\
         \cline{1-5}
         & \multicolumn{4}{c}{Verification Times (in Seconds)} \\
         \cline{1-5}
         \makecell[c]{Robust Bound \\ (Lemma~\ref{Theo:checkrobust} - Algorithm~\ref{algo:underapproximation})} & $0.0050$ & $0.0048$ & $0.0047$ & $0.0048$ \\
         \makecell[c]{Robustness Algorithm \\ (Theorem~\ref{Theo:SDP} - Algorithm~\ref{algo:robustcheck})} & $1.3260$ & $2.7071$ & $4.6285$ & $6.9095$ \\
         \cline{1-5}
    \end{tabular}}
    \caption{Verification Results of Quantum Bits Classification}
    \label{tab:QBC}
    \vspace{0cm}
\end{table}
To show the scalability of our robust bound given in Lemma~\ref{lem:bound}, we use it to develop an algorithm (Algorithm~3 in Appendix~F of the extended version of this paper~\cite{guan2020robustness}) to under-approximate the robust accuracy, which is computed by Algorithm~\ref{algo:robustcheck}. 
Algorithm~\ref{algo:underapproximation} is a subroutine of Algorithm~\ref{algo:robustcheck} without calling an SDP solver (whenever a potential non-robust state can be detected by the robust bound in Lemma~\ref{lem:bound}). We compare the verification times by Algorithms~\ref{algo:robustcheck} and~3.  

The experiments are done on a computer with the following configurations: Intel(R) Core(TM) i7-9700 CPU @ 3.00GHz $\times$ 8 Processor, 15.8 GiB Memory, Ubuntu 18.04.5 LTS, with CVXPY: Python Software for Disciplined Convex Programming~\cite{diamond2016cvxpy} for solving SDP,  and a SciPy solver for finding the minimum of constrained nonlinear multivariable function for solving QCQP.
\begin{table}
    \vspace{0cm}
    \centering
    \resizebox{0.8\linewidth}{!}{
    \renewcommand\arraystretch{1.2}
    \setlength{\arrayrulewidth}{.1em}
    \begin{tabular}{lcccc}
        \cline{2-5}
         & \multicolumn{4}{c}{Robust Accuracy  (in Percent)} \\
         \cline{2-5}
         &  $\epsilon = 0.0001$ & $\epsilon=0.0002$ & $\epsilon = 0.0003$ & $\epsilon = 0.0004$ \\
         \cline{1-5}
         \makecell[c]{Robust Bound \\ (Lemma~\ref{Theo:checkrobust} - Algorithm~\ref{algo:underapproximation})} & $99.20$ & $98.80$ & $98.60$ & $98.30$ \\
         \makecell[c]{Robustness Algorithm \\ (Theorem~\ref{Theo:SDP} - Algorithm~\ref{algo:robustcheck})} & $99.20$ & $98.80$ & $98.60$ & $98.40$ \\
         \cline{1-5}
         & \multicolumn{4}{c}{Verification Times (in Seconds)} \\
         \cline{1-5}
         \makecell[c]{Robust Bound \\ (Lemma~\ref{Theo:checkrobust} - Algorithm~\ref{algo:underapproximation})} & $1.4892$ & $1.4850$ & $1.4644$ & $1.4789$\\
         \makecell[c]{Robustness Algorithm \\ (Theorem~\ref{Theo:SDP} - Algorithm~\ref{algo:robustcheck})} & $19.531$ & $25.648$ & $28.738$ & $33.537$ \\
         \cline{1-5}
    \end{tabular}}
    \caption{Verification Results of Quantum Phase Recognition.}
    \label{tab:QPR}
     \vspace{-0.8cm}
\end{table}
The experimental results are given in Tables~\ref{tab:QBC}--\ref{tab:MNIST}.  As an example, we illustrate the details of the result for the case of $\epsilon=0.001$ in Table~\ref{tab:QBC}. First, we only apply our robust bound in Lemma~\ref{lem:bound} to pick up all potential non-robust states from the 800 points in the training dataset. Then $95$ points are left. Thus, the under-approximation of the robust accuracy computed by Algorithm~3 (in Appendix~F of the extended version of this paper~\cite{guan2020robustness}) is  88.13\%. Next, we check the  $0.001$-robustness  by Algorithm~\ref{algo:robustcheck}. Indeed, only $80$ of the points detected by the above robust bound are non-robust and the exact robust accuracy is 90.00\%. We also compare the verification time of the two approaches to the robust accuracy. See the second column in  Table~\ref{tab:QBC} for the detail, and other experiment results of $\epsilon$-robustness are also summarized in the same table. 
\begin{table}
    \vspace{-0.2cm}
    \centering
    \resizebox{0.8\linewidth}{!}{
    \renewcommand\arraystretch{1.2}
    \setlength{\arrayrulewidth}{.1em}
    \begin{tabular}{lcccc}
        \cline{2-5}
         & \multicolumn{4}{c}{Robust Accuracy  (in Percent)} \\
         \cline{2-5}
         &  $\epsilon = 0.0001$ & $\epsilon=0.0002$ & $\epsilon = 0.0003$ & $\epsilon = 0.0004$ \\
         \cline{1-5}
         \makecell[c]{Robust Bound \\ (Lemma~\ref{Theo:checkrobust} - Algorithm~\ref{algo:underapproximation})} & $99.05$ & $98.81$ & $98.21$ & $97.86$ \\
         \makecell[c]{Robustness Algorithm \\ (Theorem~\ref{Theo:SDP} - Algorithm~\ref{algo:robustcheck})} & $100.0$ & $100.0$ & $100.0$ & $100.0$ \\
         \cline{1-5}
         & \multicolumn{4}{c}{Verification Times (in Seconds)} \\
         \cline{1-5}
         \makecell[c]{Robust Bound \\ (Lemma~\ref{Theo:checkrobust} - Algorithm~\ref{algo:underapproximation})} & $1.2899$ & $1.2794$ & $1.2544$ & $1.2567$ \\
         \makecell[c]{Robustness Algorithm \\ (Theorem~\ref{Theo:SDP} - Algorithm~\ref{algo:robustcheck})} & $209.52$ & $244.79$ & $325.97$ & $365.30$\\
         \cline{1-5}
    \end{tabular}}
    \caption{Verification Results of Cluster Excitation Detection}
    \label{tab:CED}
     \vspace{0cm}
\end{table}
\begin{table}
    \vspace{0cm}
    \centering
    \resizebox{0.8\linewidth}{!}{
    \renewcommand\arraystretch{1.2}
    \setlength{\arrayrulewidth}{.1em}
    \begin{tabular}{lcccc}
        \cline{2-5}
         & \multicolumn{4}{c}{Robust Accuracy  (in Percent)} \\
         \cline{2-5}
         &  $\epsilon = 0.0001$ & $\epsilon=0.0002$ & $\epsilon = 0.0003$ & $\epsilon = 0.0004$ \\
         \cline{1-5}
         \makecell[c]{Robust Bound \\ (Lemma~\ref{Theo:checkrobust} - Algorithm~\ref{algo:underapproximation})} & $99.70$ & $99.40$ & $99.30$ & $99.20$ \\
         \makecell[c]{Robustness Algorithm \\ (Theorem~\ref{Theo:SDP} - Algorithm~\ref{algo:robustcheck})} & $99.80$ & $99.60$ & $99.30$ & $99.30$ \\
         \cline{1-5}
         & \multicolumn{4}{c}{Verification Times (in Seconds)} \\
         \cline{1-5}
         \makecell[c]{Robust Bound \\ (Lemma~\ref{Theo:checkrobust} - Algorithm~\ref{algo:underapproximation})} & $0.0803$ & $0.1315$ & $0.0775$ & $0.0811$\\
         \makecell[c]{Robustness Algorithm \\ (Theorem~\ref{Theo:SDP} - Algorithm~\ref{algo:robustcheck})} & $0.3955$ & $0.6751$ & $0.7653$ & $0.8855$ \\
         \cline{1-5}
    \end{tabular}}
    \caption{Verification Results of the Classification of MNIST}
    \label{tab:MNIST}
     \vspace{-0.8cm}
\end{table}
Tables~\ref{tab:QBC}--\ref{tab:MNIST} for the verification results show that in all of these experiments, the robust bound obtained in Lemma~\ref{lem:bound}  scales very well, and the robustness verification by Algorithm~3 costs significantly less time ($<2 s $) than the way of  computing the optimal robust bound by Algorithm~\ref{algo:robustcheck}. For example, for quantum phase recognition, for $\epsilon=0.0001,0.0002$ and $0.0003$, the under-approximation of the robust accuracy  is exactly same to the real value. Even for the last case of $\epsilon=0.0004$, only $0.1\%$ difference is got. Furthermore, from the tables, the verification time of Algorithm~\ref{algo:robustcheck} is increasing with the value of $\epsilon$, while the running time of the method by the robust bound is almost unchanged. This is because the former algorithm  uses a SDP or QCQP solver to search all possible adversarial examples for the potential non-robust states picked up by the robust bound, and the number of these states are growing up with the value of $\epsilon$.  These counter-examples detected by the algorithm confirms that our robustness framework is effective. 
For instance, see Fig.~\ref{fig.mnist.robust1} for two visualized adversarial examples generated by  Algorithm~\ref{algo:robustcheck} with a QCQP solver. As we can see, the benign and adversarial images are perceptually similar. This also proves that our robustness verification algorithm can detect not only quantum but also classical adversarial examples.

\begin{figure}[h]
    \centering    
    \includegraphics[width=0.95\linewidth]{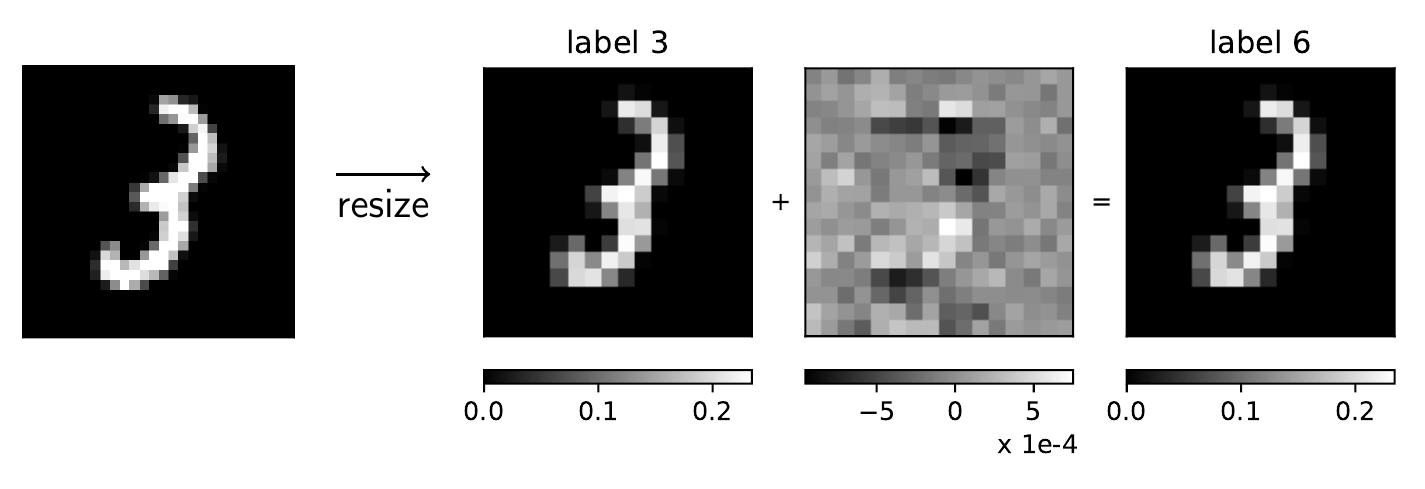}
            
    \includegraphics[width=0.95\linewidth]{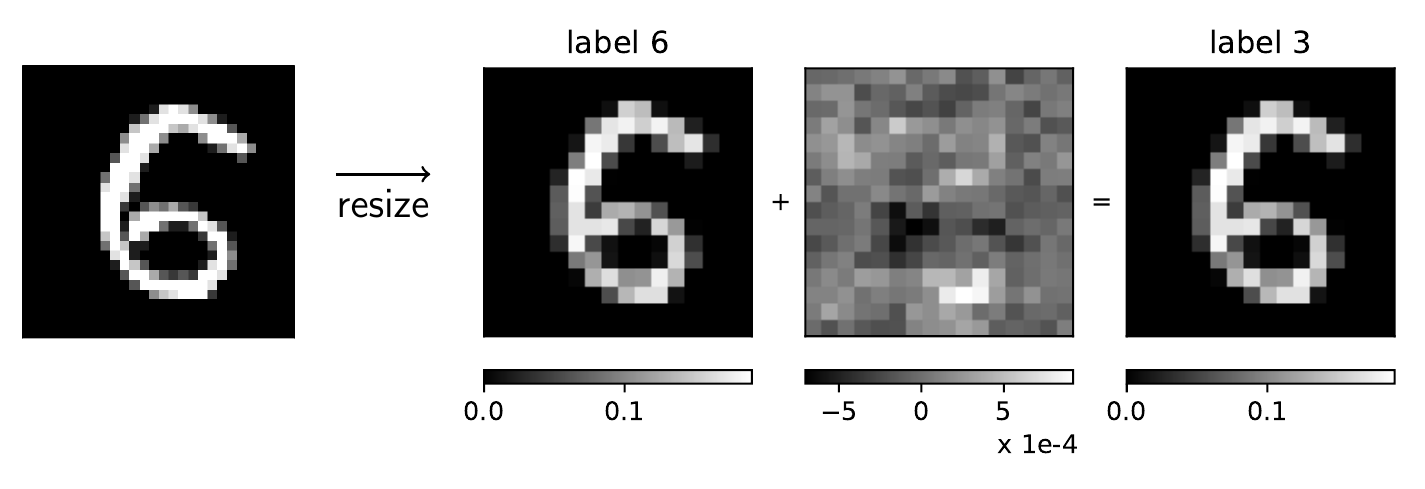}
    \caption{Two training states and their adversarial examples generated by Algorithm~\ref{algo:robustcheck} with a QCQP solver: the first column images are $28\times 28$ benign data from MNIST; The second column shows the two downscaled $16\times 16$ grayscale images; The last column images are decoded from adversarial examples founded by Algorithm~\ref{algo:robustcheck}. The third column images are the grayscale difference between benign and adversarial images.}
    \label{fig.mnist.robust1}
    \vspace{ 0 cm}
\end{figure}

\section{Conclusion}
In this work, we initiate the research of the formal robustness verification of quantum machine learning algorithms against unknown quantum noise. We found an  analytical robustness bound which can be efficiently computed to under-approximate the robust accuracy in practical applications. Furthermore, we developed a robustness verification algorithm that can exactly verify the $\epsilon$-robustness of quantum machine learning algorithms and provides useful counter-examples for the adversarial training.  

For topics for future research, it should be useful in practical applications to find an efficient method that over-approximates the robust accuracy of quantum classifiers. Combined the under-approximation approach developed in this work, it can help us to more accurately and fast estimate the robust accuracy. In classical machine learning, there exist some works in the literature to achieve this task. For instance, ImageStars, a new set representation, is introduced in ~\cite{tran2020verification} to perform efficient set-based analysis by combining operations on concrete images with linear programming, which leads to eﬃcient over-approximative analysis of classical convolutional neural networks. 

Tensor networks are one of the best-known data structures for implementing large-scale quantum classifiers (e.g. QCNNs with 45 qubits in~\cite{cong2019quantum}). 
For practical applications, we are going to incorporate tensor networks into our robustness verification algorithm so that it can scale up to achieve the demand of NISQ devices (of $\geq$ 50 qubits).

More generally, further investigations are required to better understand the role of the robustness in quantum machine learning, especially through more experiments on real world applications like learning phases of quantum many-body systems.  
\section*{Acknowledgments}
We would like to thank the anonymous reviewers for their insightful comments. This work was partly supported by the National Key R\&D Program of China (Grant No: 2018YFA0306701), the National Natural Science Foundation of China (Grant No: 61832015) and the Australian Research Council (ARC) under grant No.DP210102449.
%\begin{table*}[!htbp]
%\begin{tabular}{|c|c|c|c|c|c|c|}
%\hline
% Examples & $\#$ Training data & $\#$ Validation data &Training accuracy & Validation accuracy& $\# (p_1-p_2\leq 2\sqrt{\epsilon})$& \# Non-robust data   \\
%\hline
%QSC & 800 & 200& 100.00\% & 100.00\% & 0& 0 \\
%\hline
%QPR& 1000 & 400& 97.70\% & 95.25\%& 8 & 8 \\
%\hline
%CED  & 840 & 360 & 99.76\%& 99.46\% & 7& 0\\
%\hline
%MNIST  & 1000 & 400 & 98.40\% & 97.50\% & 3 & 3\\
%\hline
%\end{tabular}
%\caption{Training and Robustness Results of Experiments}
%\label{tab:exp}
%\end{table*}
%\emph{Discuss}---
% In our setting, the noise is purely quantum. To our best knowledge, our SDP method is the first way to automatically generating quantum adversarial examples. 
%In this scenario,   Our robust bound is effective in all quantum states. 
%Here, we show that it is highly robust to adversarial perturbations
%
% ---- Bibliography ----
% 
% BibTeX users should specify bibliography style 'splncs04'.
% References will then be sorted and formatted in the correct style.
%

%\bibliographystyle{splncs04}
\bibliographystyle{unsrt} 
\bibliography{main}
\clearpage
\section*{Appendix}
\appendix
\section{Fidelity v.s. Trace Distance}\label{appendix:fidelity_trace}
In this section, we give more details of the reason that fidelity is better than trace distance as the metric in the study of the robustness of quantum machine learning against quantum noise.

As we know, a quantum computation can be modeled by a super-operator $\e$~\cite{nielsen2010quantum}. Under quantum noise, $\e$ is disturbed to be another super-operator $\f$. How to measure the disturbance with some metric $d(\e,\f)$? We have to apply them to a quantum system prepared in a state $\rho$. Then we obtain the  optimal discrimination by maximizing over $\rho$:
\begin{equation}\label{eq:metric}
d(\e,\f)=\max_{\rho}D(\e(\rho),\f(\rho))
\end{equation}
where $D(\cdot,\cdot)$ is some metric of state distinguishability. So we need to find a good metric of state distinguishability. Of all good candidates, fidelity and trace distance are two of the most widely used metrics. 

 Trace distance quantifies the maximum probability $P_{correct}$ of correctly guessing whether $\rho$ or $\sigma$ was prepared after making a measurement~\cite{nielsen2010quantum}:
     \[P_{correct}=\frac{1}{2}+\frac{1}{2}T(\rho,\sigma)\]
     where \[T(\rho,\sigma)=\frac{1}{2}\|\rho-\sigma\|_{tr}\]
     is  the trace distance between $\rho$ and $\sigma$, and $\|\rho-\sigma\|_{tr}$ is trace norm.
     Thus when $N$ copies of $\rho$ and $\sigma$ are provided, 
     \[P_{correct}=\frac{1}{2}+\frac{1}{2}T(\rho^{\otimes N},\sigma^{\otimes N}).\]
      So how does  $T(\rho^{\otimes N},\sigma^{\otimes N})$ behave as $N$ tends to infinite as the scenario of quantum machine learning? The quantum Chernoff bound explains this~\cite{audenaert2007discriminating}:
     \[T(\rho^{\otimes N},\sigma^{\otimes N})\approx 1-\frac{1}{2}const\cdot e^{-C(\sigma,\rho)N}\] 
   where $const$ is a non-zero positive constant and $C(\sigma,\rho)$ is the quantum Chernoff bound which is generally hard to be computed. 
   And when the states are pretty close, it’s close to the infidelity (1-$F(\rho,\sigma)$)~\cite{blume2017distinguishable}:
   \[\frac{1-F(\rho,\sigma)}{2}\leq C(\rho,\sigma)\leq 1-F(\rho,\sigma).\]
   Therefore, let $\epsilon>0$ be a small enough threshold value,
   $$P_{correct}\approx 1-\epsilon \Rightarrow \frac{2\eta}{1-F(\rho,\sigma)}\geq N\geq \frac{\eta}{1-F(\rho,\sigma)},$$ 
   where $\eta=-\ln \frac{4\epsilon}{const}$.
   In terms of trace distance, we have: 
   $$P_{correct}\approx 1-\epsilon \Rightarrow \frac{\eta}{\|\rho-\sigma\|_{tr}^2}\geq N\geq \frac{2\eta}{2\|\rho-\sigma\|_{tr}-\|\rho-\sigma\|_{tr}^2} $$
   which is derived by the Fuchs-van de Graaf inequalities
   \begin{equation*}
    \begin{aligned}
        &1-\sqrt{F(\rho,\sigma)}\leq T(\rho,\sigma)\leq \sqrt{1-F(\rho,\sigma)}\\
        \Rightarrow& \frac{2\|\rho-\sigma\|_{tr}-\|\rho-\sigma\|_{tr}^2}{2}\leq C(\rho,\sigma)\leq \|\rho-\sigma\|_{tr}^2.
    \end{aligned}
   \end{equation*}

    By the above inequalities about $N$, we can see that 
    infidelity gives a better (linear) estimation of how many samples we need to accurately discriminate $\rho$ from $\sigma$ than trace distance which gives a polynomial estimation of $N$. For instance, when $\epsilon=\frac{1}{4}const\cdot e^{-100}$ (then $\eta=100$), for the same value of infidelity and trace distance, saying $0.01$, the estimations of $N$ is $[10^4,2\times 10^4]$ and $[10051,10^6]$, respectively. Thus fidelity is more suitable as the metric $D(\rho,\sigma)=1-F(\rho,\sigma)$ in Eq.(\ref{eq:metric}) in the scenario of quantum machine learning that many copies of quantum states are provided.
   
  % In other words, infidelity represents the distinguishability of humans to recognize states.

\section{Proof of Lemma~\ref{Theo:checkrobust}}\label{Appendix: proof_checkrobust}
\begin{proof}
    For any $\sigma$ with $F(\rho, \sigma) \geq 1 - \epsilon$, by the monotonicity of the fidelity~\cite[Theorem 9.6]{nielsen2010quantum}, 
    we have \[F(\e(\rho), \e(\sigma))\geq F(\rho, \sigma) \geq 1 - \epsilon.\] 
    Let
    \[\vec{p} = (\sqrt{p_1}, \sqrt{p_2}, \ldots, \sqrt{p_n})\in\mathbb{R}^n, \text{ where } p_k = \tr(M_{k}^\dagger M_k\e(\rho))\]
    \[\vec{q} = (\sqrt{q_1}, \sqrt{q_2}, \ldots, \sqrt{q_n})\in\mathbb{R}^n,  \text{ where } q_k = \tr(M_{k}^\dagger M_k\e(\sigma))\]
    Without loss of generality, we assume $p_1\geq p_2\geq \ldots \geq p_n$. We use $\vec{x}\cdot\vec{y}$ and $\norm{\vec{x}} = \sqrt{\vec{x}\cdot\vec{x}}$ to denote the inner product of $\vec{x},\vec{y}$ and the $2$-norm, respectively. Then $\vec{p}$ and $\vec{q}$  both have unit norms:\[\norm{\vec{p}} = \sqrt{\sum_k p_k} = 1 \text{ and } \norm{\vec{q}} =\sqrt{\sum_k q_k}= 1.\]  By the definition of fidelity, we have 
    \[\vec{p}\cdot\vec{q} = \sum_{k}\sqrt{p_kq_k} \geq \sqrt{F(\e(\rho), \e(\sigma))} \geq \sqrt{1-\epsilon}.\]
    
    We see that any probability  distribution  $R = (r_1, r_2, \ldots, r_n),$ with $r_k\geq 0 $ and $  \sum_k r_k = 1$ can be viewed as a unit vector $(\sqrt{r_1}, \sqrt{r_2}, \ldots, \sqrt{r_n})$ and the fidelity of two distributions is the inner product of their corresponding unit vectors.
    
    Next, we prove that the unit vector form $\vec{p}$ of $\rho$ can be used to obtain a robust bound for it.
    First, we find a vector that has maximum inner product with $\vec{p}$ and is within another class rather than the belonging class of $\rho$. This can be done by solving the following optimization problem:
    \begin{equation}\label{eq:fidelity1}
        \begin{array}{rcl}
            \max. & \vec{x}\cdot\vec{p} \\
            \mbox{s.t.} & \norm{\vec{x}} = 1 \\
            & \prod_{j=2}^n(x_1 - x_j) &= 0 \\
            & \vec{x} = (x_1, \ldots, x_n) &\in\mathbb{R}^n \\
            & \vec{x} &\geq 0
            \end{array}
    \end{equation}
    With constraint $\norm{\vec{x}} = 1$, we have $\vec{x}\cdot\vec{p} = (a\vec{x})\cdot\vec{p}/\norm{a\vec{x}}$ for any $a>0$. Thus, let $\vec{y} = a\vec{x}$, the above optimization problem is rewritten as: 
    \begin{equation}\label{eq:fidelity2}
        \begin{array}{rcl}
            \max. & {\vec{y}\cdot\vec{p}}/{\norm{\vec{y}}} \\
            \mbox{s.t.} & \norm{\vec{y}} &> 0 \\
            & \prod_{j=2}^n(y_1 - y_j) &= 0 \\
            & \vec{y} = (y_1, \ldots, y_n) &\in\mathbb{R}^n \\
            & \vec{y} &\geq 0
            \end{array}
    \end{equation}
    The objective function is not changed by multiplying the numerator and denominator with a positive number. Thus, we can assume $\vec{y}\cdot\vec{p} = 1$ and obtain the following problem
    \begin{equation}\label{eq:fidelity3}
        \begin{array}{rcl}
            \min. & {\norm{\vec{y}}}^2 \\
            \mbox{s.t.} & \vec{y}\cdot\vec{p} &= 1 \\
            & \prod_{j=2}^n(y_1 - y_j) &= 0 \\
            & \vec{y} = (y_1, \ldots, y_n) &\in\mathbb{R}^n \\
            & \vec{y} &\geq 0
            \end{array}
    \end{equation}
    By the assumption of  $p_1\geq p_2\geq \ldots\geq p_n$, we claim that the constraint 
    $$\prod_{j=2}^n(y_1 - y_j) = 0 $$
    can be replaced by $y_1 - y_2 = 0$. Suppose the optimal solution is 
    \[\vec{y}^* = (y_1^*, y_2^*, \ldots, y_j^*, \ldots, y_n^*) \text{ and } y_1^* = y_j^*,\quad j \neq 2\] 
    Then let 
    \[\vec{x}(\lambda) = \lambda \vec{p} + (1-\lambda)\vec{y}^{*},\quad  0\leq \lambda\leq 1\]
    For the case of $y_1^* \leq y_2^*$, with $p_1 \geq p_2$,  we have \[\sqrt{p_1} - \sqrt{p_2} \geq 0 \text{ and } y_1^* - y_2^* \leq 0.\]
    then there exists $0\leq \lambda_0 \leq 1$ such that \[\lambda_0(\sqrt{p_1} - \sqrt{p_2}) + (1-\lambda_0)(y_1^* - y_2^*) = 0\]
    which is equivalent to
    \[ \lambda_0\sqrt{p_1} +(1-\lambda_0)y_1^* = \lambda_0\sqrt{p_2} +(1-\lambda_0)y_2^*.\]
    So the first and second elements of $\vec{x}(\lambda_0)$ are equal. 
    
    We have
    \[\vec{x}(\lambda_0) - \vec{p} = (1-\lambda_0)(\vec{y}^* - \vec{p})\qquad
    \text{ and } \qquad(\vec{y}^* - \vec{p})\cdot \vec{p} = 0,\] which means that $\vec{y}^* - \vec{p}$ and $\vec{p}$ are orthogonal. Then 
    \begin{align*}
        \norm{\vec{x}(\lambda_0)}^2 &= \norm{\vec{x}(\lambda_0) - \vec{p} + \vec{p}}^2 \\
        &= \norm{\vec{x}(\lambda_0) - \vec{p}}^2 +\norm{\vec{p}}^2 \\
        &= (1-\lambda_0)^2\norm{\vec{y}^* - \vec{p}}^2 + \norm{\vec{p}}^2 \\
        &\leq \norm{\vec{y}^* - \vec{p}}^2 + \norm{\vec{p}}^2 \\
        &= \norm{\vec{y}^* - \vec{p} + \vec{p}}^2 \\
        &= \norm{\vec{y}^*}^2.
    \end{align*}
    We find a vector $\vec{x}(\lambda_0)$ satisfies $y_1 - y_2 = 0$ and $\norm{\vec{x}(\lambda_0)}^2 \leq \norm{\vec{y}^*}^2$.
    
    Now we consider the situation of  $y_1^* > y_2^*$. We have $y_j^* = y_1^* > y_2^*$ and $p_2 \geq p_j$, and  following the same analysis in the above,  we can find $0 < \lambda_0 \leq 1$ such that the second and $j$-th elements of $\vec{x}(\lambda_0)$ are equal and
    \begin{align*}
        \norm{\vec{x}(\lambda_0)}^2 &= \norm{\vec{x}(\lambda_0) - \vec{p} + \vec{p}}^2 \\
        &= \norm{\vec{x}(\lambda_0) - \vec{p}}^2 +\norm{\vec{p}}^2 \\
        &= (1-\lambda_0)^2\norm{\vec{y}^* - \vec{p}}^2 + \norm{\vec{p}}^2 \\
        &< \norm{\vec{y}^* - \vec{p}}^2 + \norm{\vec{p}}^2 \\
        &= \norm{\vec{y}^* - \vec{p} + \vec{p}}^2 \\
        &= \norm{\vec{y}^*}^2
    \end{align*}
    which is contradict to $\norm{\vec{y}^*}^2$ is the optimal value. Thus, the optimal value is achieved at a vector $\vec{y}$ satisfying $y_1 - y_2 = 0$, then the problem can be reformulated as
    \begin{equation}\label{eq:fidelity4}
        \begin{array}{rcl}
            \min. & {\norm{\vec{y}}}^2 \\
            \mbox{s.t.} & \vec{y}\cdot\vec{p} &= 1 \\
            & y_1 - y_2 &= 0 \\
            & \vec{y} = (y_1, \ldots, y_n) &\in\mathbb{R}^n \\
            & \vec{y} &\geq 0
            \end{array}
    \end{equation}
    Using the Lagrange multiplier method, the optimal value of problem~(\ref{eq:fidelity4}) is 
    \[\frac{2}{2 - (\sqrt{p_1}-\sqrt{p_2})^2}.\]
    Then the optimal value of problem~(\ref{eq:fidelity1}) is 
    \[\sqrt{1 - \frac{(\sqrt{p_1}-\sqrt{p_2})^2}{2}}.\]
    Therefore, if
    \[\sqrt{1 - \frac{(\sqrt{p_1}-\sqrt{p_2})^2}{2}} < \sqrt{1 - \epsilon}  \Leftrightarrow\sqrt{p_1} - \sqrt{p_2} > \sqrt{2\epsilon},\]
    then for any vector  $\vec{q}$ with $\sqrt{1-\epsilon} \leq \vec{p}\cdot\vec{q}$, we have the corresponding quantum state $\sigma$ with $F(\rho, \sigma) \geq 1 - \epsilon$ and $\sigma$ is classified into the class of $\rho$. In other words, $\rho$ is $\epsilon$-robust.
\end{proof}

\section{Proof of Theorem~\ref{Theo:checkrobust}}\label{appendix:proof_robustcheck}
\begin{proof} The sufficient direction directly follows from the definition of adversarial robustness in Definition~\ref{def:adversary}. For the necessary 
direction, if there is $k\in \c$ such that there is a solution $\sigma$ in the above problem, then $\a(\sigma)=k$, i.e., $\sigma$ is classified in the class $k$. That is $\sigma$ is an $\epsilon$-adversarial example of $\rho$ and $\rho$ is not $\epsilon$-robust. 

Now we prove that the above problem can be reduced to a SDP.  First,
it is easy to verify $\dh$, the set of quantum states on $\h$, is a convex set of positive semidefinite matrices. Computing $F(\rho,\sigma)$ can be reduced to solving a SDP~\cite{watrous2009semidefinite}. Thus replacing $F(\rho,\sigma)$ by the SDP, the above problem is a SDP problem by noting that   $\tr(\sigma)=1$ is equivalent to $\tr(\sigma)\leq 1$ and $-\tr(\sigma)\leq -1$. 
\end{proof}

\section{Pure State Robustness Verification}\label{appendix:pure_state_robustness}
In this section, we discuss the robustness verification for pure states, i.e. pure state $\ket{\psi}$ against adversarial examples of pure states.  That is that all quantum states in the training dataset and their adversarial examples are  pure states. Then, by Theorem~\ref{Theo:SDP} and noting that the set of all pure states is not convex, computing the optimal robust bound for pure states is not an SDP. But we can prove it is a  Quadratically Constrained Quadratic Program (QCQP), which is hard to be solved.

In mathematical optimization, a  QCQP is an optimization problem in which both the objective function and the constraints are quadratic functions. It has the form
\begin{align*}
    \min. \quad&\frac{1}{2} x^TP_0x+q_0^{T}x  \\
     \text{subject to} \quad& \frac 1 2 x^T P_i x + q_i^T x + r_i \leq 0, \quad \text{for } i = 1, \ldots, m \\
     & A x = b
\end{align*}
where $P_0, P_1, \ldots, P_m$ are $n\times n$ matrices and $x\in\mathbb{R}^n$ is the optimization variable. The problem is convex, if $P_0, P_1,\ldots, P_m$ are all positive semidefinite, but non-convex, if these matrices are neither positive nor negative semidefinite. In general, solving QCQP is an NP-hard problem.

\begin{corollary}[Pure State Optimal Robust Bound] \label{Theo:QCQP1}
$\a=(\e,\{M_k\}_{k\in\c})$ be a quantum classifier and $\ket{\psi}$ is a pure state with $\a(\ketbra{\psi}{\psi})=l$. Then $\delta$ is the optimal robust bound of pure state $\ket{\psi}$ against adversarial examples of pure states, where $\delta=\min_{k\not=l}\delta_k$ and $\delta_k$ is the solution of the following QCQP:
\begin{align*}
\delta_k=\min_{\ket{\varphi}\in \cH}\quad& 1 - \braket{\varphi}{\psi}\braket{\psi}{\varphi}\\
\text{subject to} \quad    &\braket{\varphi}{\varphi}=1 \\
    &\bra{\varphi}\cE^{\dagger}(M_{l}^\dagger M_l-M_k^{\dagger}M_k)\ket{\varphi}\leq 0 
\end{align*}
where if the QCQP is unsolved, then  $\delta_k=+\infty.$
\end{corollary}
\begin{proof}
First, we note that $F(\ketbra{\psi}{\psi},\ketbra{\varphi}{\varphi})=\braket{\varphi}{\psi}\braket{\psi}{\varphi}$ and $\tr(\ketbra{\varphi}{\varphi}A)=\bra{\varphi}A\ket{\varphi}$ for any matrix $A$. Then the result follows from Theorem~\ref{Theo:SDP} and the observation that $\braket{\varphi}{\varphi}=1$  is equivalent to
\[ \bra{\varphi}I\ket{\varphi}-1\leq 0 \textrm{ and } -\bra{\varphi}I\ket{\varphi}+1\leq 0,\]
where $I$ is the identity matrix on $\h$.
\end{proof}

As we can see, the above QCQP is non-convex, so we cannot efficiently compute the pure state optimal robust bound like the mixed state case by SDP.  However, there are some numerical tools designed to solve QCQP---not only compute the minimal value but also output a corresponding optimal solution $\ket{\varphi}$, as SDP solvers. Some methods have been developed to approximately solving non-convex QCQPs in a reasonable time. For example,  non-convex QCQPs with non-positive off-diagonal elements can be exactly solved by SDP relaxations~\cite{kim2003exact}. Therefore, there may have a polynomial-time algorithm to solve this specific form of QCQP, and we left this problem as a future research.

\section{Training Models}\label{appendix:traing_models}
\begin{figure}
\centering
\subfloat[]{\label{fig.qcnn.diagram}
      \centering
      \parbox[][3cm][c]{0.34\linewidth}{
        \includegraphics[width=\linewidth]{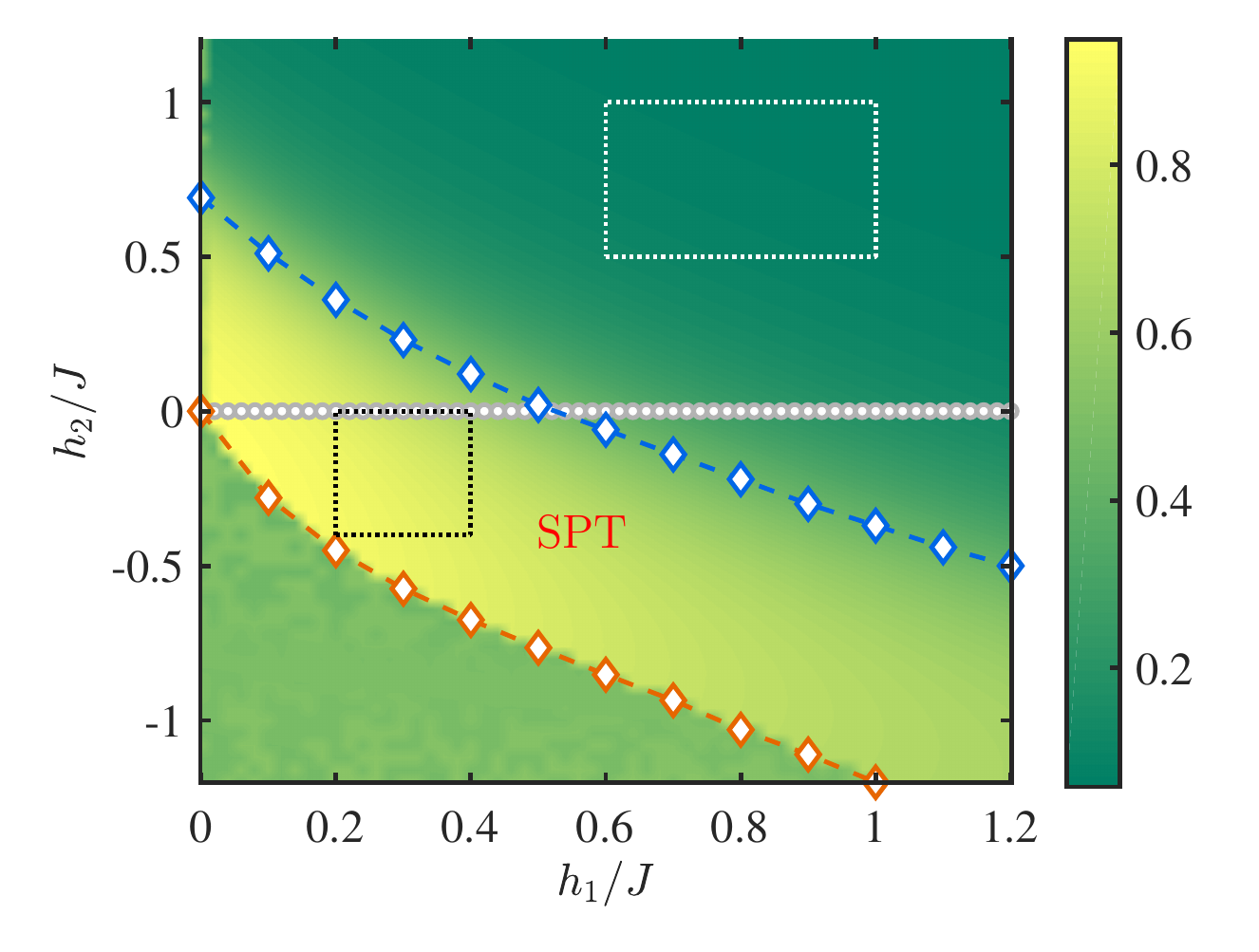}
      }
    }
\subfloat[]{\label{fig.qcnn.model}
      \centering
      \parbox[][3cm][c]{0.6\linewidth}{
        \includegraphics[width=\linewidth]{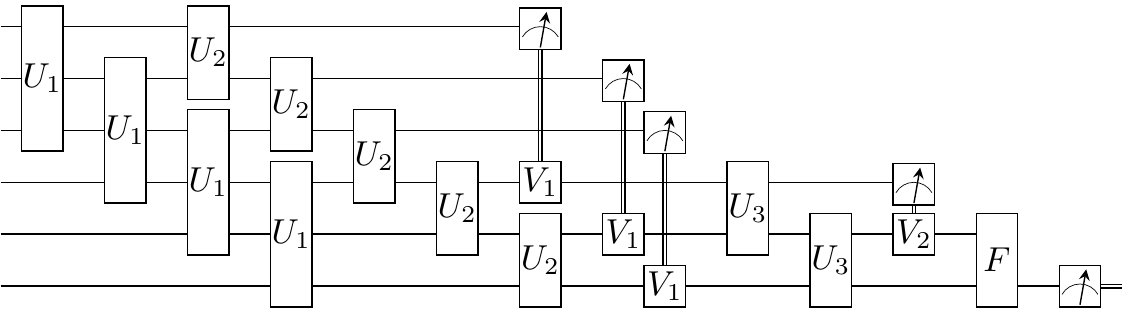}
      }
    }
%\subfloat[]{\label{fig.qcnn.loss}
%      \centering
%      \includegraphics[width=0.3\linewidth]{figures/phase_loss.pdf}
%    }
%\subfloat[]{\label{fig.qcnn.accuracy}
%      \centering
%      \includegraphics[width=0.3\linewidth]{figures/phase_accuracy.pdf}
%    }
    
\caption{ (a) The phase diagram obtained by our trained QCNN model for input size $N=6$ spins. (b) Our  QCNN circuit model. \iffalse(c) The MSE loss for our training and retraining (+8) processes.  (d) Accuracy for our training and retraining (+8) processes.\fi }
\end{figure}

\begin{figure}[htbp]
    \centering
    \subfloat[]{\label{fig.excitation.generator}
        \centering
        \includegraphics[width=0.205\linewidth]{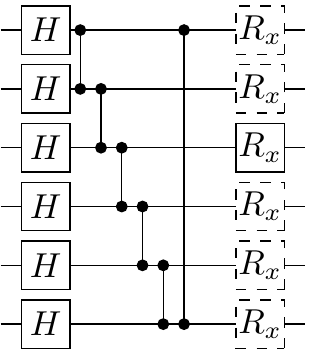}
    }
    \subfloat[]{\label{fig.excitation.model}
        \centering
        \includegraphics[width=0.74\linewidth]{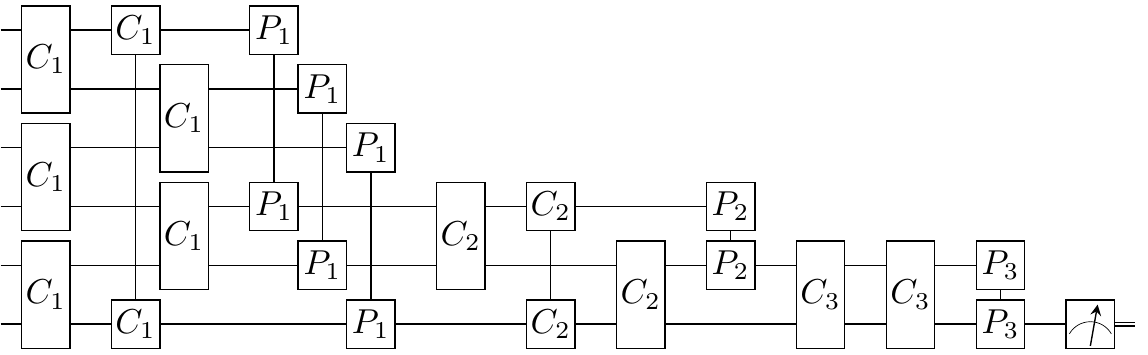}
    }
    \caption{(a) The circuit generating cluster state. (b) The classification model for cluster excitation detection.}
\end{figure}

\begin{figure}
    \centering
    \includegraphics[width=0.8\linewidth]{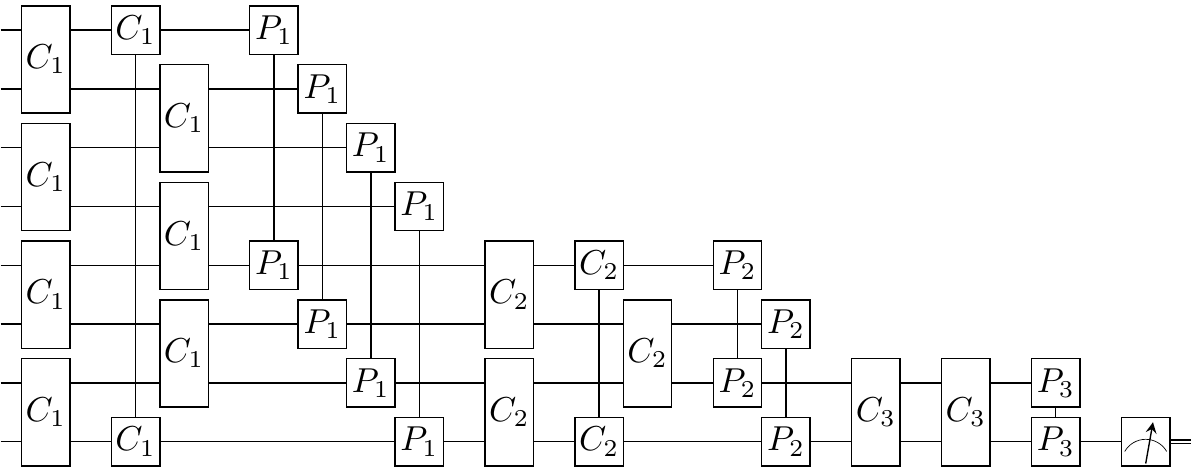}
    \caption{QCNN model  for the classification of MNIST}
    \label{fig.mnist.model}
\end{figure}
\section{Under-approximation Algorithm of Robust Accuracy}\label{appendix:under_robust_accuracy}
\begin{algorithm}[H]
\caption{UnderRobustAccuracy($\a,\epsilon,T$)}
\label{algo:underapproximation}
    \begin{algorithmic}[1]
    \Require  $\a=(\e,\{M_k\}_{k\in\c})$ is a well-trained quantum classifier, $\epsilon < 1$ is a real number, $T=\{(\rho_i,l_{i})\}$ is the training dataset of $\a$
    \Ensure A  under-approximation of robust accuracy URA
    \State $r=0$.\Comment{Record the number of potential non-robust states}
    \ForAll{$(\rho_i,l_i)\in T$}
    \State Let $p_1$ and $p_2$ be the first and second largest elements of $\{\tr(M_{k}^\dagger M_k\e(\rho_i))\}_k$, respectively.
    \If {$\sqrt{p_1}-\sqrt{p_2}\leq \sqrt{2\epsilon}$} \Comment{Applying the robust bound in Lemma~\ref{lem:bound}}
    %\State \textbf{Skip}
    %\Else
    \State $r=r+1$
    \EndIf
    \EndFor
    \State \Return URA$=1-\frac{r}{|T|}$.
    \end{algorithmic}  
\end{algorithm}

\end{document}